\documentclass[11pt]{article}

\newif\ifsubmission
\submissionfalse
\usepackage{fouriernc}
\usepackage{braket}
\usepackage[margin=1in]{geometry}
\usepackage{amsmath,amssymb,amsthm,mathtools}
\usepackage{microtype}
\usepackage{enumitem}
\usepackage[hidelinks]{hyperref}
\usepackage{cleveref}
\usepackage{xcolor}
\usepackage{authblk}
\usepackage{graphicx}
\usepackage{algorithm}
\usepackage{algpseudocode}
\algtext*{EndFor}
\algtext*{EndIf}
\usepackage{tikz} \usetikzlibrary{arrows.meta,angles,quotes,positioning,calc,decorations.pathreplacing}
\usepackage{quantikz}

\newtheorem{theorem}{Theorem}[section]

\newtheorem{lemma}[theorem]{Lemma}
\newtheorem{corollary}[theorem]{Corollary}

\theoremstyle{definition}
\newtheorem{definition}[theorem]{Definition}
\theoremstyle{remark}

\newcommand{\fbs}{\operatorname{fbs}}
\newcommand{\OR}{\operatorname{OR}}
\newcommand{\AND}{\operatorname{AND}}
\newcommand{\E}{\mathbb{E}}

\newcommand{\cF}{\mathcal{F}}

\newcommand{\cB}{\mathcal B}
\newcommand{\cH}{\mathcal H}
\providecommand{\cF}{\mathcal F}
\providecommand{\cG}{\mathcal G}
\providecommand{\cR}{\mathcal R}

\newcommand{\eps}{\varepsilon}

\newcommand{\deph}{\mathrm{deph}}
\newcommand{\negq}{\mathrm{neg}}

\newcommand{\bit}{\{0,1\}}
\newcommand{\grover}{\textsf{Grover}}

\newcommand{\PiQ}[1]{\Pi^{Q}_{#1}}
\newcommand{\PiC}[1]{\Pi^{C}_{#1}}
\newcommand{\PiA}[1]{\Pi^{A}_{#1}}

\newcommand{\ketbra}[1]{\ket{#1}\bra{#1}}%
\counterwithin{equation}{section}

\begin{document}
\title{Noisy Quantum Query Complexity via Fractional Block Sensitivity}
\date{}

\ifsubmission
  \author{}
\else
  \author{
  Mehil Agarwal,
  Shravas Rao, and
  Fang Song \\
  Department of Computer Science, Portland State University, USA\\
  \texttt{\{mehil,shravas,fang.song\}@pdx.edu}
}
\fi

\maketitle
\thispagestyle{empty}

\begin{abstract}
We study quantum query complexity under several models of imperfect oracle access, and develop lower bounds through a common framework based on fractional block sensitivity (\(\fbs\)).

For a negligent oracle that applies the correct query with probability \(1-p\), we prove a lower bound in terms of \(\fbs(f,0^n)\). We also show, perhaps surprisingly, that negligence need not destroy quantum speedups: any function \(f\) can be transformed into a partial function \(f'\) whose negligent query complexity essentially preserves the quantum query complexity of \(f\). This gives partial functions with exponential quantum speedups even under negligent queries, and in particular rules out a general lower bound in terms of \(\fbs(f)\) for partial functions in this model.

For two other models, we obtain general lower bounds in terms of \(\fbs(f)\) for all Boolean functions. For hybrid algorithms using \(Q\) coherent and \(C\) classical queries, we prove the tradeoff \(C+Q^2=\Omega(\fbs(f))\). For an IID dephasing noisy model where each query dephases the query-index register at rate \(p\) independently, we prove \(\Omega\!\left(p\,\fbs(f)\right)\) queries are necessary.

Finally, we introduce a broader family of time-varying dephasing models and identify a variational resource that is always lower bounded by \(\fbs(f)\). Computing this resource reduces to a convex optimization problem, providing a simple way to derive lower bounds for new noise schedules. As applications, we recover the hybrid and IID dephasing bounds and determine the query complexity of unstructured search when the dephasing rate grows over time.

\end{abstract}

\newpage
\setcounter{page}{1}
\section{Introduction}
\label{sec:intro}
Quantum computers are capable of solving certain problems more efficiently than their classical counterparts. A canonical example is that of unstructured search, in which one is asked to find a marked item out of $n$ possible items.
A classical algorithm must search through at least a majority of the list, and hence requires $\Omega(n)$ queries to the input.
On the other hand, there exists a quantum algorithm that uses $O(\sqrt{n})$ quantum queries to the input by Grover's algorithm~\cite{G96}, and this algorithm is optimal~\cite{BBBV97, BBHT98, Ambainis02}.

It is interesting to ask, to what extent this speed-up holds in more limited models of quantum computation.
Currently, quantum computers that can solve unstructured search in its full generality do not yet exist.
However, it is conceivable that in the way to building working quantum computers, one will have access to quantum computers with various limitations in the near future.
And thus, the question becomes, can one still obtain a quadratic speed-up for unstructured search?
And can other problems lend themselves to a speed-up in these limited models?

In this work, we consider several models.
The first was studied by Regev and Schiff~\cite{RS08} in which the quantum query fail with some probability $p$, which we refer to as the negligent oracle model.
The second was studied by Rosmanis~\cite{Ansis22}, in which there is a limited number of quantum queries allowed and otherwise an algorithm only gets to issue classical queries, which we refer to as the hybrid model.
The third was also studied by Rosmanis~\cite{R23}, in which each quantum query is dephased with some probability $p$.
In all three settings, it was shown that speed-up over classical algorithms no longer exists.
That is, not only does Grover's algorithm no longer solve the problem in $O(\sqrt{n})$ queries, but no algorithm can.
In all three models, $\Omega(n)$ queries are required, as for classical algorithms.
In the hybrid and dephasing models, these results were further refined and simplified in~\cite{HLS24}.

In addition, we also introduce a fourth model that can be seen as a generalization of both the hybrid model and the dephasing model described above.
In this model, each query is dephased with some probability that can vary over the queries, and in fact the dephasing between queries does not need to be independent.

We generalize the three results above to show that quantities related to the fractional block sensitivity, a relaxation of block sensitivity independently introduced by Tal~\cite{T13} and  Gilmer, Saks and Srinivasan~\cite{GSS16}, is a lower bound for the query complexity in all four models.
In particular, the fractional block sensitivity at $0^n$ is a lower bound for the negligent model, and the maximum fractional block sensitivity over all inputs are lower bounds for the other three models.
As the fractional block sensitivity of unstructured search at $0^n$ is $n$, this recovers all three of the results in~\cite{RS08,Ansis22,R23}.

The significance of fractional block sensitivity is that it accurately captures the randomized query complexity of many problems.
In particular, it was shown that fractional block sensitivity is equivalent to many other lower bounds that have been introduced for randomized query complexity~\cite{AKPV17}.
These include Aaronson's relational adversary bound~\cite{A06}, Ambainis's weighted adversary bound~\cite{Ambainis06}, the Kolmogorov complexity adversary bound~\cite{LM08}, and the minimax over probability distributions adversary bound.
Thus, one interpretation of the results in this paper is that for many problems, namely those for which any of these lower bound techniques capture the right lower bound, the presence of a faulty oracle or a limited amount of quantum queries removes any possibility of a quantum speed-up.

Finally, we note that in the case of the negligent oracle model, it does seem necessary to only consider the fractional block sensitivity at $0^n$.
In particular, we show that for any function $f$, there is a partial function $f'$ with the same randomized query complexity, but whose query complexity in the negligent oracle model is the same as the quantum query complexity of $f$.

\subsection{Main Results}

We start by defining fractional block sensitivity, which is essential to all of our lower bounds.
\label{def:fbs-lp}
Given a boolean function $f: \{0, 1\}^n \rightarrow \{0, 1\}$.
The fractional block sensitivity at $x$ that is $\mathsf{fbs}(f, x)$ is defined by the optimum value of the following linear program.
\begin{alignat}{3}
 & \text{maximize}  &\sum_{B: f(x) \neq f(x^B)} w_B & \label{ref:lp}\\
 & \text{subject to} \quad &\sum_{B\ni i} w_B  
 &\leq 1,  \quad \text{for }i = 1 ,..., n \nonumber\\ 
                 && 0 \leq w_B &\leq 1,\quad \text{for } B\subseteq [n] \nonumber.
\end{alignat}
Here, \(x^B\) is the input obtained by flipping block \(B\). Note that the value of the corresponding integer program in which the variables $w_B$ are required to be $0$ or $1$ is equivalent to the block sensitivity at $x$.
The fractional block sensitivity $\mathsf{fbs}(f)$ is the maximum of $\mathsf{fbs}(f, x)$ over all choices of $x$.
\paragraph*{The standard quantum query model}

We review the standard quantum query model. 
Here one has query access to the input $x \in  \{0, 1\}^n$ in the form of an oracle 
\begin{equation} \label{eq:standard-qo}
    O_x : \ket{i}\ket{b} \mapsto \ket{i}\ket{b\oplus x_i}.
\end{equation}
The equivalent quantum phase oracle is defined as
 \begin{equation} \label{eq:phase-qo}
        O_x : \ket{i}\ket{b} \mapsto (-1)^{bx_i}\ket{i}\ket{b}.
 \end{equation}

A quantum query algorithm with query complexity $T$ is a sequence of unitary operators $U_0...U_T$ interleaved with queries to the quantum oracle $O_{x}$ with input $x.$ After $t \in \{0,...,T\}$ queries on input $x$, the state of the algorithm $\ket{\phi^t_x}$ is defined as
    \[
    \ket{\phi^t_x} = U_t O_x U_{t-1} O_x... U_1 O_x U_0 \ket{0}
    \]
    The algorithm succeeds in computing a function $f$ if it outputs the correct value on any input with probability at least $1-\eps$. The quantum query complexity $Q_{\eps}(f)$ of function $f$ is defined as the smallest query complexity among the algorithms that compute that function.

In the noisy models we consider in this paper, we use modifications of the oracle that capture the relevent differences.

\paragraph*{The negligent oracle model}

In the negligent oracle model of Regev and Schiff, the oracle works as intended with some fixed probability $1-p$, and does nothing with probability $p$.
That is,
\begin{equation}\label{eq:faultyo}
\mathcal N_{p,x}(\rho) = (1-p)O_x(\rho) + p \rho.
\end{equation}

Letting $Q^{\negq}_{p,\eps}(f)$ denote the number of queries needed to compute $f$ on any input with probability at least $1-\eps$ in the negligent oracle model, we obtain the following lower bound

\begin{theorem}[Negligent-oracle lower bound]
\label{thm:neg-anchor}
For the negligent oracle with  noise rate $p$,
\[
  Q^{\negq}_{p,\eps}(f)
  =\Omega_\eps\!\left(\frac{p}{1-p}\,\fbs(f,0^n)\right).
\]
\end{theorem}

\paragraph*{The hybrid model}In the hybrid model of Rosmanis, one allows a certain number, $Q$, of quantum queries and a certain number, $C$, of classical queries.
Similarly to Rosmanis, we will upper bound the power of classical queries by so-called pseudo-classical queries, a type of query introduced in~\cite{AHU19}.
We let such a query to input $x$ be defined as
\begin{equation}\label{eq:pco}
\mathcal{P}_x(\rho) = \sum_{b = 0, 1} \left(\sum_{i: x_i = b} \ket{i}\bra{i} \otimes \ket{b}\right) \rho \left(\sum_{i: x_i = b} \ket{i}\bra{i} \otimes \bra{b}\right).
\end{equation}
We obtain the following lower bound.

\begin{theorem}[Hybrid-oracle lower bound]\label{thm:hybrid}
Any bounded-error algorithm computing $f$ with probability at least $1-\epsilon$ using $C$ pseudo-classical queries and $Q$ coherent quantum
queries satisfies
\[
  C+Q^2=\Omega_\eps(\fbs(f)).
\]
\end{theorem}

We note again that as classical queries are weaker than pseudo-classical queries, the same lower bound also holds for an algorithm that makes both classical and quantum queries.

\paragraph*{IID dephasing model}

In the IID dephasing model of Rosmanis, each query is dephased with probability $p$.
In particular, each query evolves the state as
    \begin{equation} \label{eq:iid-dephasing-o}
    \mathcal{D}_{p,x}
  (\rho)\longmapsto
    (1-p)O_x\rho O_x^\dagger +  
    p\Delta_I(O_x\rho O_x^\dagger)
    \end{equation}
    where
    \[
  \Delta_I(\rho)=\sum_iJ_i\rho J_i\ \text{ and } J_i = |i\rangle \langle i | \otimes I \, .
  \]

Letting $Q^{\deph}_{p,\eps}(f)$  denote the number of queries needed to compute $f$ on any input with probability at least $1-\eps$ in the iid dephasing model, we obtain the following lower bound.

\begin{theorem}[IID query-index dephasing oracle lower bound]
\label{thm:iid-intro}
Suppose each query is independently fully dephased in the query-index register with probability $p$.
Then every bounded-error algorithm for $f$ requires

\[
  Q^{\deph}_{p,\eps}(f)
  =\Omega_\eps\!\left(p\,\fbs(f)\right).
\]
\end{theorem}

It is not too hard to see that the fractional block sensitivity of unstructured search on $n$ items is $n$ on the input $0^n$.
Thus, Theorem~\ref{thm:neg-anchor}, Theorem~\ref{thm:hybrid}, and Theorem~\ref{thm:iid-intro} can be seen as generalizations of the results in~\cite{RS08},~\cite{Ansis22} and~\cite{R23} respectively.

One interpretation of these results is that when the number of quantum queries is limited to at most $O(\sqrt{\mathsf{fbs}(f)})$ in the hybrid model, or when the dephasing probability is constant in the dephasing model, one obtains a lower bound of $\Omega(\mathsf{fbs}(f))$ to successfully compute the function $f$.
As stated previously, this is the same lower bound as one obtains from the many different classical adversary methods.
Thus, for many problems besides that of unstructured search, it is no longer possible to achieve a speed-up in these models.

In the case of the negligent oracle, one can ask if it is possible to obtain a lower bound of $\fbs(f)$, as opposed to $\fbs(f, 0^n)$.
The answer to this is no if one allows $f$ to be a partial function.
And in fact for every function $f$, one can construct a partial function $f'$ that has the same randomized query complexity, but whose query complexity when using the negligent oracle is the same as the quantum query complexity of $f$.
Thus, the same quantum speed-ups that can occur in the standard quantum query model can also occur in the negligent model as long as we allow for partial functions in the latter.

\begin{theorem}[No global-\(\fbs\) bound for partial functions] \label{thm:neg-partial-lifting}
For every function $f: \{0, 1\}^{n} \rightarrow \{0, 1\}$, there exists a partial function $f': \{0, 1\}^{2n} \rightarrow \{0, 1\}$ such that $R(f) = R(f')$ and
\[
Q^{\negq}_{p,\eps}(f')
  =\Omega_\eps\!\left(\frac{1}{1-p}\,Q_{\eps+1/10}(f)\right).
\]
\end{theorem}

\paragraph{Extended dephasing models.} The IID dephasing model assumes the same dephasing rate at every query.  More generally, we can allow the dephasing pattern to be drawn from an arbitrary distribution \(\mu\) that dictates whether dephasing occurs at each query and may be may be correlated across time. Let \( (Z_1, \ldots, Z_T) \sim \mu \) be indicator random variables and 
\[p_t=\Pr_\mu[Z_t=1\mid Z_1,\ldots,Z_{t-1}] \]
be the dephasing probability at time \(t\), conditioned on the past schedule. We show that an opportunity resource function with respect to \(\mu\) on a set of variables \(\theta=(\theta_t)_{t=1}^T\) necessarily exceeds \(\fbs(f)\), in order for a \(T\)-query algorithm to correctly compute \(f\).    

\begin{theorem}[Computational-Opportunity Theorem]
\label{thm:opportunity} If a $T$-query algorithm computes
$f$ with worst-case error $\eps<1/2$ under the dephasing schedule $\mu$, then
\[
  \cR(\mu): = \inf_\theta \E_\mu\!\left[
      \left(\sum_{t=1}^T\theta_t\right)^2
      +\sum_{t=1}^T\frac{(1-\theta_t)^2}{p_t}
    \right] \ge c_\eps\fbs(f).
\]
\end{theorem}

The significance of this theorem is that it separates the quantum part of the argument from the noise schedule.  Once the theorem is proved, analyzing a new dephasing model no longer requires a new quantum lower-bound proof.  One only needs to estimate the scalar optimization problem \(\cR(\mu)\).  For example, we can easily recover the lower bounds in hybrid and IID depahsing models.  To further demonstrate it, we consider a scenario that dephasing rate increases over time reflecting the challenge of increasing difficulty of preserving coherence over long computation on realistic quantum device. Formally, let \[
  p_t=\min\{1,\lambda t^\gamma\},
  \qquad 0<\gamma<1,
  \qquad \lambda>0.
  \label{eq:pht-powerlaw-profile}
\]
where \(\lambda\) controls the overall noise scale and
\(\gamma\) controls the rate at which the noise worsens with time. The \(\mu\) then samples dephasing error at each query with rate \(p_t\). Our computational opportunity theorem readily gives a tight bounds for computing OR. 

\begin{theorem}[Increasing power-law dephasing for search]
\label{thm:powerlaw}
For every fixed $\eps<1/2$ and $0<\gamma<1$, in the flagged model,
\[
Q^{(p_t)}_\eps(\OR_N)
=\Theta_{\eps,\gamma}
\begin{cases}
\sqrt N,
&\lambda\lesssim N^{-(1+\gamma)/2},\\[1mm]
(\lambda N)^{1/(1-\gamma)},
&N^{-(1+\gamma)/2}\lesssim\lambda\lesssim N^{-\gamma},\\[1mm]
N,
&\lambda\gtrsim N^{-\gamma}.
\end{cases}
\]
\end{theorem}

\subsection{Technical Overview}
\label{sec:technical-overview}

Our lower bounds start from the same simple object: a weighted star of inputs arising from fractional block sensitivity.  Fix an input \(x^\star\), and let \(\{x^B\}_{B\in\mathcal B}\) be sensitive inputs obtained by flipping blocks \(B\), with feasible weights \(w_B\) satisfying
\[
    W:=\sum_{B}w_B,
    \qquad
    \sum_{B\ni i}w_B\le 1
    \quad\text{for every }i.
\]

\begin{figure}[ht]
\centering
\resizebox{0.35\linewidth}{!}{%
\begin{tikzpicture}[
    every node/.style={font=\small},
    center/.style={circle, draw, thick, minimum size=6mm, inner sep=1pt},
    leaf/.style={circle, draw, thick, minimum size=6mm, inner sep=1pt},
    wlabel/.style={fill=white, inner sep=1pt, font=\scriptsize},
    annot/.style={align=left, font=\small}
]

\node[center] (c) at (0,0) {$0^n$};

\node[leaf] (b1) at (90:2) {$x^{B_1}$};
\node[leaf] (b2) at (25:2.4) {$x^{B_2}$};
\node[leaf] (b3) at (-25:2.4) {$x^{B_3}$};
\node[leaf] (b4) at (-90:2) {$x^{B_4}$};
\node[leaf] (bm) at (155:2.4) {$x^{B_m}$};

\node at (125:2) {$\cdots$};
\node at (-155:2) {$\cdots$};

\draw (c) -- (b1) node[midway, right, wlabel] {$w_{B_1}$};
\draw (c) -- (b2) node[midway, above, wlabel] {$w_{B_2}$};
\draw (c) -- (b3) node[midway, below, wlabel] {$w_{B_3}$};
\draw (c) -- (b4) node[midway, right, wlabel] {$w_{B_4}$};
\draw (c) -- (bm) node[midway, above left, wlabel] {$w_{B_m}$};
\end{tikzpicture}
}
\caption{Weighted-star input structure induced by a fractional block-sensitivity packing at the center $0^n$. Sensitive blocks \(B_1,\dots,B_m\) at \(0^n\); \(\sum_{B\ni i} w_B \le 1 \quad (\forall i\in[n]) \); \(W=\sum_B w_B=\fbs(f,0^n)\). }
\label{fig:weighted-star}
\end{figure}
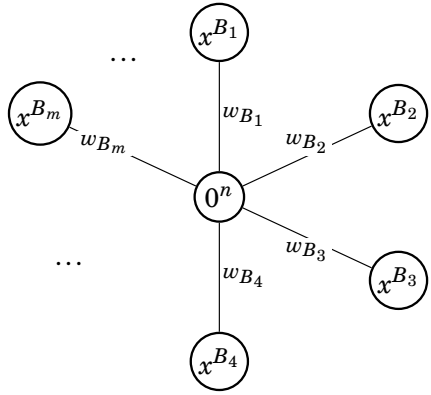

At optimal choice of \(w_B\), \(W=\fbs(f,x^\star)\).  Writing \(\pi_B=w_B/W\), the constraint becomes \(    \sum_{B\ni i}\pi_B\le \frac{1}{W}\). Thus, no single query index can impact more than a \(1/W\) fraction of the weighted leaves. This elementary inequality is the common source of all our bounds.  At the amplitude level it gives a scale \(1/\sqrt W\): a coherent query can create at most \(O(1/\sqrt W)\) new separation between the computation on the center and the the leaves.  The different oracle models differ in what subsequently happens to this separation.  For negligent and pseudo-classical queries, it is convenient to compare a common branch of the center and leaf computations directly.  For query dephasing, one must additionally keep track of the hidden index records created by dephasing.  We describe the two arguments below. 

\paragraph{Negligent queries.}
We first consider a weighted star centered at \(x^\star=0^n\), with sensitive leaves \(x^B\) and weights \(w_B\) as in Figure~\ref{fig:weighted-star}.
On input \(x\) the negligent query channel\(\mathcal N_x(\rho)
    =
    (1-p)\,O_x\rho O_x^\dagger
    +
    p\,\rho\) has a standard Stinespring dilation 
\[
    V_x
    =
    \sqrt{1-p}\,O_x\otimes|s\rangle
    +
    \sqrt p\,I\otimes|f\rangle,
\]

where \(|s\rangle\) and \(|f\rangle\) are orthogonal environment states recording whether the oracle was applied.

For \(x^\star = 0^n\), \(O_{x^\star}=I\), and hence 
    \[V_{x^\star}
    =
    I\otimes
    \bigl(
        \sqrt{1-p}\,|s\rangle+\sqrt p\,|f\rangle
    \bigr). \]

We therefore rotate the environment basis so that \(
    |0\rangle_E
    :=
    \sqrt{1-p}\,|s\rangle+\sqrt p\,|f\rangle \)
is precisely the environment state produced by the  computation on \(0^n\).  Choosing an orthogonal vector \(|1\rangle_E\), the dilation for a leaf \(x^B\) can be written as
\[
    V_{x^B}
    =
    K_{B,0}\otimes|0\rangle_E
    +
    K_{B,1}\otimes|1\rangle_E, \quad \text{with} \quad K_{B,0}=I-2(1-p) \Gamma_B,
    \qquad \Gamma_B = \sum_{i\in B} \ket{i}\bra{i} \otimes I.
\]

This reveals the \emph{common branch} we will track, in which every environment register is in the common state \(|0\rangle_E\).  Let \(
    |\phi_t^\star\rangle \text{ and }
    |\phi_t^B\rangle
\) denote the corresponding states immediately before the \(t\)-th oracle call for the center and leaf \(x^B\), respectively.  The leaf state \(|\phi_t^B\rangle\) is generally subnormalized, since amplitude may have leaked into orthogonal environment branches.  The actual leaf density operator contains this branch as a positive component.  On the common branch, the \(t\)-th oracle call acts as
    
\[ |\phi_t^\star\rangle
    \longmapsto
    |\phi_t^\star\rangle, \quad \text{and} \quad 
    |\phi_t^B\rangle
    \longmapsto
    \left(I-2(1-p)\Gamma_B\right)|\phi_t^B\rangle.
\]

Afterwards the algorithm applies the same input-independent unitary \(U_{t+1}\) which preserves Euclidean distance.  Therefore we introduce a progress that track the normalized weighted separation

\[
    H_t
    :=
    \sum_B \pi_B
    \bigl\|
        |\phi_t^\star\rangle-|\phi_t^B\rangle
    \bigr\|^2
    =
    \frac1W
    \sum_Bw_B
    \bigl\|
        |\phi_t^\star\rangle-|\phi_t^B\rangle
    \bigr\|^2 \, .
\]

Let \( a_t^2  := \sum_B \pi_B \|\Gamma_B\phi_t^B\|^2\) represent the weighted mass of the part of the leaf computations currently querying their sensitive blocks. We can then bound the single-step growth by 
\[ H_{t+1}- H_t    \le \frac{4(1-p)}{\sqrt W}a_t - 4p(1-p)\,a_t^2 \, ,\]
Intuitively, the first term captures the separation created by the oracle action on the currently sensitive part of the query state.  The second reflects leakage from the common environment branch into the orthogonal branch \(K_{B,1}\). 

Maximizing over \(1 \ge a_t >0\) hence gives \(H_{t+1}- H_t \le \frac{1-p}{pW}\). Since \( H_0=0\) and correctness requires constant terminal separation \( H_T=\Omega_\eps(1)\), we obtain that 
    \[ Q^{\negq}_{p,\eps}(f)
    =
    \Omega_\eps\!\left(
        \frac{p}{1-p}\,
        \fbs(f,0^n)
    \right).
\]

The restriction to the \(0^n\) is essential to this argument.  The useful common environment branch exists because the center oracle is the identity. 

\paragraph{Hybrid coherent and pseudo-classical queries.}
A closely related argument applies when an algorithm can mix coherent quantum queries with pseudo-classical queries.  A pseudo-classical query records the oracle answer while preserving coherence among query indices that give the same answer. This is stronger than classical queries which are equivalent to measuring the query index completely.

Consider center \(x^* = 0^n \) again, and the pseudo-classical oracle always records answer \(0\). We therefore track the environment branch corresponding to recorded answer \(0\). For a leaf \(x^B\), this branch acts as \(I-\Gamma_B\), while the complementary branch \(\Gamma_B\) records that the query has hit the sensitive block \(B\). The same potential \(H_t
    =   \frac1W
    \sum_Bw_B
    \bigl\|
        |\phi_t^\star\rangle-|\phi_t^B\rangle
    \bigr\|^2 \) 
can be bounded universally for every query step: 
\begin{align*}
   \text{pseudo-classical: } & \quad   \Delta H_{\rm pc}
    \le
    \frac{2a}{\sqrt W}-a^2
    \le
    \frac1W, \\
    \text{quantum: } & \quad \Delta H_{\rm q}
    \le
    \frac{4}{\sqrt W}.
\end{align*}

Thus an algorithm using \(Q\) coherent queries and \(C\) pseudo-classical queries must satisfy \(
    \frac{Q}{\sqrt W}+\frac{C}{W}
    =
    \Omega_\eps(1)\), and hence
\[
    C+Q^2=\Omega_\eps(W) \, .
\]

This matches the expected interpolation: coherent queries can accumulate progress at the amplitude scale \(1/\sqrt W\), whereas de-cohered queries make only \(O(1/W)\) progress per use.

\paragraph{Index dephasing.} 

The preceding common-branch argument is no longer sufficient for query-index dephasing.  We therefore switch from viewing the weighted star merely as a classical prior distribution to working with its coherent purification in a truth register,

\[     \frac1{\sqrt2}
    \left(
        |0^n\rangle
        +
        \sum_B\sqrt{\pi_B}\,|x^B\rangle
    \right).
\] 

Query-index dephasing reveals which query index was used, which brings an additional complication.  To retain this information in the analysis, we purify the dephasing channel as well, storing the dephased query index in an inaccessible environment register.  The resulting joint truth--algorithm--record state allows us to decompose the truth space according to the accumulated hidden index record and to track separately the unresolved direction, coherent distinguishing directions, and leaves already hit by the record.

Let \(R\) denote the set of indices recorded so far and define \(\cB_R := \{B:B\cap R=\varnothing\}\), which include the leaves not yet hit by the record.  Inside the truth register, we split the state into three parts:
\begin{align*}
A_R &: \text{ the common direction of the center and unresolved leaves},\\
Q_R &: \text{ coherent separation among the unresolved leaves},\\
C_R &: \text{ leaves already certified as hit by the hidden record}.
\end{align*}
Let \(q_t\) and \(c_t\) denote the squared magnitudes in the \(Q\)- and \(C\)-parts, and we will track the progress by a potential function \(\Psi_t:=c_t+3q_t\).

For a query to address \(i\), there is one particular direction inside \(Q_R\) that interacts directly with \(A_R\).  The overlap between \(A_R\) and this direction is \(O(1/\sqrt W)\). Let \(a_t\) be the magnitude of the component in this active direction.  We can analyze the effect of each type of query on the potential.

Roughly, a coherent query can bring in a fresh distinguishing component of magnitude \(O(1/\sqrt W)\) from \(A_R\), contributing \(O(1/W)\) to the squared-norm potential; it can also interfere with an already present active component of magnitude \(a_t\), contributing \(O(a_t/\sqrt W)\). This hence would give

\[
    \Delta_{q}\Psi_t  \lesssim \frac{1}{W} + \frac{a_t}{\sqrt W}.
\]

A dephased query has an additional effect.  Once the queried index is recorded, the old active \(Q\)-direction disappears: its unhit part is absorbed into the new common unsolved direction, while its hit part moves to \(C\).  This decreases the potential by \(O(a_t^2)\), and the growth in this case is roughly 

\[
    \Delta_{c}\Psi_t
    \lesssim
    \frac1 W
    +
    \frac{a_t}{\sqrt W}
    - a_t^2.
\]

Since each query is independently dephased with probability \(p\), averaging the two cases yields

\[
    \Delta\Psi_t
    \lesssim
    \frac{1}{W}
    +
    \frac{a_t}{\sqrt W}
    -
    p\,a_t^2 \lesssim O(\frac{1}{p W}).
\]
This then implies the \( p\fbs(f,0^n) \) lower bound.

\paragraph{Lifting hybrid and dephasing lower bounds to general \(\fbs\).} 

We can observe that dephasing on index commutes with XOR translation of the input oracle. This means that in both hybrid (with classical, or quantum queries) and dephasing models, given oracle \(O_x\) we can simulate \(O_{x \oplus x^*}
\) without any overhead. As a result, the \(\fbs(f,0^n)\) lower bounds also holds against \(\fbs(f,x^*)\), and the optimal \(x^*\) hence gives \(\fbs\) the lower bound. 

\paragraph{Failure of global \(\fbs\) lower bound or lifting quantum speedup into negligent model.} However, the XOR translation fails for negligent queries, because the oracle becomes identity in the event of negligence but XOR still gets applied and hence deviates from the desired course of action.  This turns out to be more than merely a technical defect of the proof. For partial functions, a universal lower bound in terms of global \(\fbs(f)\) is \emph{false}.

Starting from a function \(f\) on \(n\) bits, we define a partial function \(f'\) based on it that can be computed in the negligent model using no more than a small overhead on the number of queries to compute \(f\) in the standard noiseless model.  The trick is to intentionally embed a flag bit to signal whether negligence occurred without disturbing coherent computation. Consequently, every ideal query can be simulated by repeating negligent queries until success which incurs only a \(\frac1{1-p}\) overhead in expectation. 

To achieve this,  we encode every bit into a complementary pair
\[
    0\longmapsto 01,
    \qquad
    1\longmapsto 10.
\]

Let \(f'\) be defined only on such strings \( \{(01)\cup (10)\}^*\) and set \(f'(x')=f(x)\) for \(x'=(x,b)\). Then the phase oracle factors as \(O_{x'}=O_x\otimes Z\), and we can view the two coordinates of \(x'\) as a logical query index together with a one-qubit flag.  Hence preparing the flag in \(|+\rangle\), and it will be flipped to \(\ket{-} = Z\ket{+}\) only if the negligent oracle succeeds. Measuring the flag therefore reveals \emph{exactly} whether the query succeeded, without revealing anything about the logical computation.

As a result, starting with a function \(f\) whose \(\fbs\) exceeds its quantum query complexity, we obtain a partial function \(f'\) computable with fewer queries than the fractional block sensitivity, which refutes a universal \(\fbs\) lower bound in the negligent model.  
On the positive side, however, this construction also shows that any partial function \(f\) having a large separation between quantum and randomized query complexity, \(f'\) retains essentially the same quantum advantage even in the negligent model, offering abundant examples of quantum speedup. 

In retrospect, the effect of negligence depends not only on how many sensitive directions the function has, but also on whether the oracle action can itself reveal that a query succeeded. This is precisely the feature that is absent at the center \(0^n\) and hence enables the lower-bound proof. 

\paragraph{General dephasing schedules.} The potential analysis by the subspace partition in IID dephasing can be extended to depict a similar basic one-step picture: a query can create coherent separation at scale \(1/\sqrt W\), while dephasing suppresses the currently active part of that separation at a rate proportional to \(p_t\).  The difficulty now is that the strength of this suppression varies with time and with the realized past schedule. 

To capture this tradeoff, the predictable parameters  \(\theta_t\in[0,1] \) are hence introduced. They are chosen as a function of the previous dephasing events.  Intuitively, it determines how the opportunity at time \(t\) is accounted for.  A large \(\theta_t\) treats the query as an opportunity for coherent accumulation, while a small \(\theta_t\) takes advantage of the available dephasing to control that progress more strongly. This leads to formulating the opportunity resource

\[    \cR(\mu)
    :=
    \inf_\theta
    \E_\mu\!\left[
        \left(\sum_{t=1}^T\theta_t\right)^2
        +
        \sum_{t=1}^T
        \frac{(1-\theta_t)^2}{p_t}
    \right] \, .
\]

The first term reflects coherent accumulation which adds at the amplitude level, leading to a square of their total weight.  The second term reflects the extent to which dephasing can be used to control the remaining progress where stronger dephasing makes this accounting cheaper.  Relating it to the potential function growth will establish our Computational-Opportunity Theorem states that any bounded-error algorithm must satisfy \(\cR(\mu)\ge c_\eps\,\fbs(f)\). 

It might be instructive to mention why another intuitive characterization fails. When a query algorithm is run with a dephasing oracle, we may partition the execution into phases of coherent runs separated by dephased queries. This suggests a natural quantity \(C+\sum_j \ell_j^2\), where \(C\) is the number of dephased queries and \(\ell_j\) is the length of each realized coherent-run phase. This is adequate for deterministic schedules but fails for random ones. It evaluates the schedule with hindsight, where a rare long coherent run receives quadratic credit even though its length is not known when the queries in that run are made. The correct resource must instead be predictable, depending only on the past schedule through the conditional dephasing probabilities \(p_t\). This alludes to working with the predictable process \(\theta_t\). 

\subsection{Discussion and Open Questions}
\label{sec:open}
Our results leave several natural directions open.  In the negligent model, our lower bound holds at \(\fbs(f,0^n)\). This reflects an intrinsic asymmetry of the model. The identity failure branch coincides with \(0^n\). If failed queries are replaced by a fixed reference oracle \(O_a\), the same argument yields a bound in terms of \(\fbs(f,a)\). Thus the theorem can be interpreted as a reference-point lower bound. For some canonical detection problems this entails no loss. For example, \(\fbs(f,0^n)=\fbs(f)\) for OR, parity, and Tribes in a suitable parameter regime. However, the partial-function construction shows that a universal lower bound in terms of global \(\fbs(f)\) is false in general.  It remains open whether \emph{total} Boolean functions can nevertheless admit such a \(\fbs\) bound.  More broadly, \(\fbs\) is not always expected to give the tight bound.  It would be interesting to determine the negligent query complexity of functions such as \(\AND\circ\OR\), where the standard quantum complexity is much larger than \(\sqrt{\fbs}\), and to identify a sharper complexity measure for this model.

In the dephasing model, a first concrete problem is again to determine the dephasing complexity of \(\AND\circ\OR\).  This would test whether the \(\fbs\)-based lower bound is close to optimal or whether a stronger adversary-type quantity is needed.  A second direction is robust simulation.  Recent work~\cite{LMP24} shows that arbitrary quantum query algorithms can be simulated with polynomial overhead in the negligent model, preserving sufficiently large quantum speedups. It would be interesting to obtain an analogous general simulation theorem for query dephasing, ideally with overhead matching or approaching our lower bounds.

Finally, the Computational-Opportunity Theorem applies to general predictable dephasing schedules, but the explicit examples we analyze are still generated by independent dephasing events at each time step.  An important direction is to apply the theorem to more physically motivated correlated schedules, such as burst noise, recovery processes, or environments with temporal memory.  Understanding which such models lead to qualitatively new opportunity resources may help connect query-complexity bounds more directly with realistic noise processes.

\section{Preliminaries}
\label{sec:prelim}

\subsection{Translation covariance in hybrid and dephasing models}
\label{sec:translation}
\begin{lemma}[Translation covariance of quantum and classical queries]
\label{lem:hybrid-translation-cq}
For every known bit string $a$,
\[
  O_{x\oplus a}=O_xO_a = O_aO_x,
\]
Consequently the quantum is invariant under XOR translation of the input, allowing an
arbitrary maximizer of fractional block sensitivity to be moved to $0^n$. The same holds true for a classical query.

\end{lemma}

\begin{lemma}[Translation covariance of query-index dephasing]
\label{lem:deph-translation}
For every known bit string $a$,
\[
  O_{x\oplus a}=O_xO_a,
  \qquad
  O_a \Delta_I O_x (\rho) = \Delta_I O_{x\oplus a}  (\rho).
\]
Consequently the dephasing query complexity is invariant under XOR translation of the input, allowing an
arbitrary maximizer of fractional block sensitivity to be moved to $0^n$.
\end{lemma}

\begin{proof}
    Given oracle \(O_x\) we have
    \[ 
    O_a :=  
    \sum_{i\in[n]}
    \ketbra{i}\sum_{b\in\bit}
    \ket{b\oplus a_i}\bra{b} ,
    \qquad
    O_{x \oplus a} := 
    \sum_{i\in[n]}
    \ketbra{i}\sum_{b\in\bit}
    \ket{b\oplus x_i\oplus a_i}\bra{b} .
    \]
    The index-query dephasing oracle
    \[
        \Delta_I (\rho) \mapsto 
        \sum_{i} 
        \left( J_i 
         \rho 
         J_i \right),
         \qquad
         J_i =\ketbra{i}\otimes I.
    \]
    To show $O_a \Delta_I O_x (\rho) = \Delta_I O_{x\oplus a}  (\rho)$, 
    \begin{figure}[ht]
\centering
\resizebox{0.7\linewidth}{!}{%
\begin{quantikz}[row sep=0.5cm, column sep=0.6cm]
\lstick{}
    & \gate[2]{O_x}
    & \gate{\Delta_I}
    & \gate[2]{O_a}
    & \qw
\\
\lstick{}
    &
    & \qw
    &
    & \qw
\end{quantikz}
\qquad $=$ \qquad
\begin{quantikz}[row sep=0.5cm, column sep=0.6cm]
\lstick{}
    & \gate[2]{O_x}
    & \gate[2]{O_a}
    & \gate{\Delta_I}
    & \qw
\\
\lstick{}
    &
    &
    & \qw
    & \qw
\end{quantikz}%
}
\end{figure}
\\
The left hand side,
\begin{align*}
\rho
&\xrightarrow{\,O_x\,}
\left(
\sum_{i\in[n]} \ketbra{i}
\otimes
\sum_{b\in\bit}\ket{b\oplus x_i}\bra{b}
\right)
\rho\, O_x^\dagger
\\
&\xrightarrow{\,\Delta_I\,}
\sum_{i\in[n]}
\left(
\ketbra{i}\otimes
\sum_{b\in\bit}\ket{b\oplus x_i}\bra{b}
\right)
\rho\,
(O_x J_i)^\dagger
\\
&\xrightarrow{\,O_a\,}
\sum_{i\in[n]}
\left(
\ketbra{i}\otimes
\sum_{b\in\bit}\ket{b\oplus x_i\oplus a_i}\bra{b}
\right)
\rho\,
(O_x J_i O_a)^\dagger .
\end{align*}
For the right hand side
\begin{align*}
\rho
&\xrightarrow{\,O_{x\oplus a}\,}
\left(
\sum_{i\in[n]} \ketbra{i}
\otimes
\sum_{b\in\bit}\ket{b\oplus x_i\oplus a_i}\bra{b}
\right)
\rho\, O_{x\oplus a}^\dagger
\\
&\xrightarrow{\,\Delta_I\,}
\sum_{i\in[n]}
\left(
\ketbra{i}\otimes
\sum_{b\in\bit}\ket{b\oplus x_i\oplus a_i}\bra{b}
\right)
\rho\,
(O_{x\oplus a}J_i)^\dagger .
\end{align*}
    Thus the XOR operation commutes with the dephasing oracle. Thus a query to rotated oracle can be made without extra overhead when given access to a query-index dephasing oracle.
\end{proof}

\subsection{Definitions and preliminary lemmas}
 
We start by defining quantities and proving lemmas that we will use to prove both Theorem~\ref{thm:neg-anchor} and Theorem~\ref{thm:hybrid}.
Because the proofs of the two are similar, we collect them together in order to streamline the presentation.

We define $\rho_t^B$ to be the state of the density operator on input $(0^n)^B$ (that is, the input whose $1$-indices are exactly the set $B$) after applying $t$ unitaries and $t$ query operators, which can be either quantum, negligent, or pseduo-classical.
Rather than analyzing $\rho_t^B$ directly, we will instead analyze the vectors $\ket{\phi_t^B}$ and $\ket{\widetilde{\phi}_t^B}$ which are subnormalized pure states that we define below.

 \begin{definition}
    For $B\subseteq [n]$ and $t\in \{0,\ldots, T-1\}$ we define the following vectors 
    \[
    \ket{\phi_{0}^B} \text{ is the initial state, }
    \qquad 
    \ket{\phi_t^B} = U_t\ket{\widetilde{\phi}_t^B}.
    \]
        \item Let $\ket{\phi_t^B} = \sum_{i=1}^n \ket{i, \alpha_{t,i}^B}$. Then, 
        \[
        \ket{\widetilde{\phi}_{t+1}^B} := 
        \begin{cases}
        \sum_{i=1}^n (-1)^{\mathbf{1}_{i \in B}}\ket{i, \alpha_{t,i}^B}, & \text{if the } (t+1) \text{st query is quantum,} \\
        \sum_{i=1}^n \ket{i, \alpha_{t,i}^B} - 2(1-p) \sum_{i \in B} \ket{i, \alpha_{t,i}^B},  & \text{if the } (t+1) \text{st query is negligent,} \\
        \sum_{i=1}^n \ket{i, \alpha_{t,i}^B} - \sum_{i\in B} \ket{i, \alpha_{t,i}^B}, & \text{if the } (t+1) \text{st query is pseudo-classical}.
        \end{cases}
        \]
    \end{definition}

We have that for every $t$ and $B$, $\rho_t^B = \ket{\phi_t^B}\bra{\phi_t^B}+\sigma_t^B$ for some positive semi-definite matrix $\sigma_t^B$, using the same reasoning as Lemma 1 in~\cite{RS08} and as Section 4 in~\cite{Ansis22}.
Thus, instead of working with the density matrices $\rho_B^t$, we will work with the pure states $\ket{\phi_B^t}$ and $\ket{\widetilde{\phi}_{t+1}^B}$.

We also define the following two helper vectors to make the analysis cleaner. They will help us capture the differences in an algorithm's execution when working on the input $0$ and when working on the input $(0^n)^B$.

\begin{definition}
    For $B \subseteq [n]$ and $t\in \{0,\ldots, T\}$, consider $\ket{\phi_t^B} = \sum \ket{i, \alpha_{t,i}^B}$. We define
      \begin{align*}
          \ket{\psi_t^B} = \sum_{i \in B} \ket{i, \alpha_{t,i}^B} \, , \qquad \ket{\psi_{t,B}^0} = \sum_{i \in B} \ket{i, \alpha_{t,i}^0} \, .
      \end{align*}
    \end{definition}

To prove our lower bounds (Theorem~\ref{thm:neg-anchor} and Theorem~\ref{thm:hybrid}), we define a progress measure below. It involves variables characterizing the fractional block sensitivity of our function $f$, which enables proving our query lower bounds based on the fractional block sensitivity.

\begin{definition}
    Given a function $f: \{0, 1\}^n \rightarrow \{0, 1\}$ such that $f(0) = 0$, and a solution to the linear program when $x = 0^n$~\eqref{ref:lp}, we define the progress measure
    \[
    H_t^B = \left\|\ket{\phi_t^0}-\ket{\phi_t^B}\right\|^2 ,
    \qquad
    H_t = \sum_{B: f((0^n)^B) =1}\pi_B H^{B}_t ,
    \qquad 
    H_0 = 0 ,
    \]
    where 
    \[
    \pi_B = \frac{w_B}{W},
    \qquad
    W = \sum_{B: f((0^n)^B) =1}w_B.
    \]
\end{definition}
The following lemma argues that in order to compute $f$ successfully, the progress measure $H_T$ needs to be at least $(1-4\sqrt{\epsilon})$.

\begin{lemma}[Terminal progress bound]\label{lem:progresslb}
    Suppose that for all $y \in \{0, 1\}^n$ the algorithm returns the correct answer with probability at least $1 - \epsilon$. Then $H_T \geq (1-4\sqrt{\epsilon})$. 
    \end{lemma}
    \begin{proof}
        For all $B\text{ such that } f((0^n)^B)=1$ we have 
        \begin{align*}
            H^B_T &= \left\|\ket{\phi_T^0}-\ket{\phi_T^B}\right\|^2 \\
            &= \left\|\ket{\phi_T^0}\right\|^2 + \left\|\ket{\phi_T^B}\right\|^2 - 2\text{Re}(\langle\phi_T^0|\phi_T^B\rangle) \\
            &\geq 1-2|\langle\phi_T^0|\phi_T^B\rangle|.
        \end{align*}
        Let the density operator $\rho^B_T$ be the final state of the algorithm when run on input $(0^n)^B$. We have $\rho_T^0 = \ket{\phi_T^0}\bra{\phi_T^0}$ and $\rho_T^B = \ket{\phi_T^B}\bra{\phi_T^B} + \sigma_{T}^B$ for some positive semi-definite matrix $\sigma_t^B$.
        Assuming the correctness of the algorithm that it distinguishes between the yes and no instances of $f$ with error probability at most $\epsilon$ we have 
        \begin{align*}
             1-2\epsilon &\leq \left\|\rho_T^B - \rho_T^0\right\|_{\text{tr}}\\
             &= \left\|\rho_T^B - \ket{\phi_T^0}\bra{\phi_T^0}\right\|_{\text{tr}}\\
             &\leq \sqrt{1-\langle \phi_T^0|\rho_T^B|\phi_T^0\rangle}\\
             &=\sqrt{1-\ket{\phi_T^0}(\ket{\phi_T^B}\bra{\phi_T^B}+\sigma_T^B)\ket{\phi_T^0}} \\
             &\leq \sqrt{1-|\langle \phi_T^0|\phi_T^B\rangle|^2},
        \end{align*}
        where the second inequality follows from the fact that for a mixed state $\rho$, one has $\|\rho - \ket{\phi}\bra{\phi}\|_{\text{tr}} \leq \sqrt{1- \langle\phi | \rho | \phi\rangle}$ (see~\cite[Chapter~9]{NC10}).
        Thus, $|\langle\phi_T^0 |\phi_T^B\rangle|\leq 2\sqrt{\epsilon}$ and $H_T^B \geq 1-4\sqrt{\epsilon}$. \,\,Finally, by the definition $H_T$, 
        \begin{align*}
            H_T &= \sum_{B: f((0^n)^B)=1}\pi_BH_T^B \\
            &\geq \sum_{B: f((0^n)^B)=1}\pi_B(1-4\sqrt{\epsilon}) \\
            &= (1-4\sqrt{\epsilon}).
        \end{align*}
        The final equality follows from the definition of fractional block sensitivity.
    \end{proof}

\section{Negligent Oracles}
\label{sec:negligent}

\subsection{A Regev--Schiff-style weighted-star proof}

In this section, we prove Theorem~\ref{thm:neg-anchor}.
We start by bounding the increase in progress measure $H_{t+1}^B-H_t^B$.
We note that this lemma closely follows that of Lemma 3 in~\cite{RS08}.
However, there are slight differences due to the fact that $B$ might have more than $1$ element, and so for completeness, we include the proof here.

\begin{lemma}\label{lem:faultinc}
For all $B\text{ such that } f((0^n)^B)=1$ and all $t \in \{0,\ldots,T-1\}$, \[H_{t+1}^B - H_t^B \leq \frac{1-p}{p}\left\|\ket{\psi_{t, B}^0}\right\|^2.\]
\end{lemma}
\begin{proof} By the definition of $H_{t+1}^y$ we have
\begin{align*}
H_{t+1}^B &= \left\|\ket{\phi_{t+1}^0}-\ket{\phi_{t+1}^B}\right\|^2 \\
&= \left\|U_{t+1}\ket{\widetilde{\phi}_{t+1}^0}-U_{t+1}\ket{\widetilde{\phi}_{t+1}^B}\right\|^2 \\ 
&= \left\|\ket{\phi_t^0}-\ket{\phi_t^B} + 2(1-p) \ket{\psi_t^B}\right\|^2 \\
&= H_t^B-4(1-p)\left\|\ket{\psi_{t}^B}\right\|^2+4(1-p)^2\left\|\ket{\psi_{t}^B}\right\|^2 +2(1-p)\langle{\psi_t^B}|{\psi_{t, B}^0}\rangle+2(1-p)\langle{\psi_{t,B}^0}|{\psi_{t}^B}\rangle \\
&\leq H_t^B-4p(1-p)\left\|\ket{\psi_t^B}\right\|^2+4(1-p)\left\|\ket{\psi_t^B}\right\|\left\|\ket{\psi_{t, B}^0}\right\| \\
&\leq H_t^B+\frac{1-p}{p} \left\|\ket{\psi_{t, B}^0}\right\|^2.
\end{align*}
The last inequality holds in maximizing the quadratic expression over $\left\|\ket{\psi_{t}^B}\right\|$.
\end{proof}

Now we are ready to obtain an upper-bound on the increase of the progress measure $H_{t+1}-H_t$.
In particular, the progress measure can not increase too much because of the constraints set on the weights $w_B$ in the linear program defined by Eq.~\eqref{ref:lp}.

\begin{lemma}[One-step growth: negligent]\label{lem:faultHt} For all $t\in\{0,\ldots T-1\}$,
     \[H_{t+1}-H_{t} \leq \frac{1-p}{pW},\]
     where $W = \sum_{B: f((0^n)^B) =1}w_B$.
 \end{lemma} 
 \begin{proof} Using $\pi_B = \frac{w_B}{W}$ and by the definition of $H_t$,
    \begin{align*}
        H_{t+1}-H_{t} &= \sum_{B: f((0^n)^B)=1} \pi_B (H^{B}_{t+1}-H^B_t) \\
        &\leq  \sum_{B: f((0^n)^B)=1} \pi_B\left(\frac{1-p}{p}\right) \left\|\ket{\psi_{t, B}^0}\right\|^2 \\        &=  \sum_{B: f((0^n)^B)=1} \pi_B \left(\frac{1-p}{p}\right) \left(\sum_{i\,:\,i \in B}\left\|\alpha_{t, i}^0\right\|^2\right) \\
        &= \frac{1-p}{pW} \sum_{i=1}^n \left\|\alpha_{t, i}^0\right\|^2\left(\sum_{B:i \in B}w_B\right) \\  
        &\leq \frac{1-p}{pW}.
    \end{align*}
 The first inequality follows from Lemma~\ref{lem:faultinc}. The second inequality follows from the fact that for a fixed $t$, $\sum_i \|\alpha_{t, i}^0\|^2 \leq 1$ and that $\sum_{B:i \in B}w_B \leq 1$. 
 \end{proof}

Finally, we combine Lemma~\ref{lem:progresslb} and~\ref{lem:faultHt} to prove Theorem~\ref{thm:neg-anchor}.

 \begin{proof}[Proof of Theorem~\ref{thm:neg-anchor}]
 By Lemma~\ref{lem:faultHt} we have \[H_T\leq \sum_{t=0}^{T-1}\frac{1-p}{pW} \,=\, \frac{1-p}{pW}T.\]
 Then that by Lemma~\ref{lem:progresslb}, and that $W=\fbs(f,0^n)$ we have
 \[T\geq \frac{p(1-4\sqrt{\epsilon})}{1-p}\,\mathsf{fbs}(f, 0^n)\]
 as desired.
 \end{proof}

\paragraph{Failure of translation covariance.}
The negligent-oracle model is not invariant under the standard XOR translation used to move an arbitrary
fractional-block-sensitivity maximizer to $0^n$: after correcting by the known translation oracle, the
fault branch is generally a nontrivial known oracle rather than the identity channel.

\subsection{Partial-function counterexample}

Finally, we prove Theorem~\ref{thm:neg-partial-lifting}.

\begin{proof}[Proof of Theorem~\ref{thm:neg-partial-lifting}]
For each input $x \in  \{0, 1\}^n$ to $f$, we identify with an input $x' \in \{0, 1\}^{2n}$ to $f'$ as follows.
For each $i$, we set $x'_{2i-1}=0 x'_{2i} = 1$ if $x_i = 0$ and $x'_{2i-1}=1, x'_{2i} = 0$ if $x_i = 1$.
Finally, we set $f'(x')$ to be $f(x)$, if $f(x)$ is defined.
For all other inputs, we leave the value of $f'$ to be undefined.

It is easy to see that any randomized algorithm for $f'$ can be turned into a randomized algorithm for $f$ with the same query complexity, and vice versa.
In particular, every query for $f'$ can be simulated by a query for $f$, and every query for $f$ can be simulated by a query for $f'$.

Let $Z = \begin{bmatrix} 1 & 0 \\ 0 & -1 \end{bmatrix}$.
Then it follows that $O_{x'} = O_x \otimes Z$ for every $x$.
It turns out that this implies that on the inputs for which $f$ is defined, the result of the negligent oracle can tell us whether or not the query was applied.
It will be helpful to divide the space into the logical register, which $O_x$ operates on, and the flag register, which $Z$ operates on.

We can obtain an algorithm for $f'$ using negligent oracles from a quantum query algorithm for $f$ as follows.
Each unitary $U$ applied between oracle calls gets applied again to the logical register.

Each time we make a query to the input in the algorithm to compute $f$, we do the following.
We prepare the flag register in the state $\ket{+}=\frac{\ket{0}+\ket{1}}{2}$, and apply the negligent oracle.
We then measure the flag register, and measure the output in the basis $\{\ket{+}, \ket{-}\}$ (where $\ket{-}=\frac{\ket{0}-\ket{1}}{2}$).
If the measurement yields $\ket{-}$, the oracle succeeded and the algorithm continues.
Otherwise, if measurement yields $\ket{+}$, the oracle failed and we repeat until the oracle succeeds.

The algorithm to compute $f'$ succeeds with probability $1-(\eps+1/10)$ as long as the corresponding algorithm to compute $f$ succeeds with probability $1-(\eps+1/10)$, using an expected $\frac{1}{1-p}Q_{\eps+1/10}(f)$ queries.
This does not exceed $\frac{10}{1-p}Q_{\eps+1/10}(f)/(1-p)$ with probability at most $9/10$ by Markov's inequality.
Thus, stopping this process after $\frac{10}{1-p}Q_{\eps+1/10}(f)/(1-p)$ steps yields an algorithm that succeeds with probability $\eps$ as desired.
\end{proof}

Theorem~\ref{thm:neg-partial-lifting} demonstrates that there are some functions for which the introduction of the negligent oracle model does not increase the number of queries required to compute the function.
Strikingly, this means that even in the negligent oracle model, there are partial functions that can be computed using $O(1)$ queries, but require $\Omega(N^{1-\eps})$ queries in the randomized query model~\cite{BS21,SSW23}.

\section{Hybrid Quantum-Classical Queries}
\label{sec:hybrid}

In this section, we prove Theorem~\ref{thm:hybrid}, our lower bound for the hybrid query model.

Recall that the hybrid query model, the hybrid algorithm is split into pseudo-classical and quantum queries. We let $C$ be the number of pseudo-classical queries and $Q$ be the number of quantum queries, such that $C + Q = T$ where $T$ is the total number of queries made by the algorithm. 
We also let $T_c \subseteq \{1,\ldots,T\}$ be the set of queries that are pseudo-classical and let $T_q := \{1,\ldots,T\} \backslash T_c$ be the set of queries that are quantum. 
Note that $|T_c| = C$ and $|T_q| = Q$.

For each $B$, define
\[
P^B_0 = \sum_{i: i \not\in B} \ket{i}\bra{i} \text{ and } P^B_1 = \sum_{i: i \in B} \ket{i}\bra{i}.
\]

For the pseudo-classical queries, on input $(0^n)^B$ the oracle acts as
\[P^B = P^B_0 \otimes \ket{0}+P^B_1 \otimes \ket{1}.\]  

We let $\ket{\phi_{t}^B}$ evolve as 
\begin{align*}
    \ket{\widetilde{\phi}^B_{t+1}} = P^B_0\ket{\phi_t^y} =\ket{\phi_t^B} - \ket{\psi_t^B}.
\end{align*}

For quantum queries, the oracle acts as usual, and in particular
\[O_B = \sum_{i} (-1)^{\mathbf{1}_{i \in B}}\ket{i}\bra{i},\]  
and the vector $\ket{\phi_{t}^B}$ evolves as 
\begin{align*}
    \ket{\widetilde{\phi}^B_{t+1}} = O_B\ket{\phi_t^y} = \ket{\phi_t^B} - 2\ket{\psi_t^B}.
\end{align*}

Below we bound the progress with respect to a pseudo-classical query (Lemma~\ref{lem:hybridp}) and a quantum query (Lemma~\ref{lem:hybridq}) respectively. 
Part of our proof closely follows Claim 5 in \cite{Ansis22}.


\begin{lemma}[One-step growth: pseudo-classical]
\label{lem:hybridp}
    For all $B \subseteq [n]$ and $t\in T_c$ we have \[H_{t+1}^B - H_{t}^B \leq \left\|\ket{\psi_{t,B}^0}\right\|^2 \,\,\,\, \text{ and } \,\,\,\, H_{t+1} - H_{t} \leq \frac{1}{W}\]
    where $W = \sum_{B: f((0^n)^B) =1}w_B$.
\end{lemma}
\begin{proof} By the definition of progress measure,
\begin{align*}
H_{t+1}^B &= \left\|\ket{\phi_{t+1}^0}-\ket{\phi_{t+1}^B}\right\|^2 \\
&= \left\|U_{t+1}\ket{\widetilde{\phi}_{t+1}^0}-U_{t+1}\ket{\widetilde{\phi}_{t+1}^B}\right\|^2 \\ 
&= \left\|P^0_0 \ket{\widetilde{\phi}_{t+1}^0}-P^B_0 \ket{\widetilde{\phi}_{t+1}^B}\right\|^2  \\
&= \left\|\ket{\phi_t^0}-\ket{\phi_t^B} +  \ket{\psi_t^B}\right\|^2 \\
&= H_t^B-2\|\ket{\psi_{t}^B}\|^2+\|\ket{\psi_{t}^B}\|^2+2\langle{\psi_t^B}|{\psi_{t, B}^0}\rangle+\langle{\psi_{t,B}^0}|{\psi_{t}^B}\rangle \\
&\leq H_t^B-\|\ket{\psi_t^B}\|^2+\|\ket{\psi_t^B}\|\|\ket{\psi_{t, B}^0}\| \\
&\leq H_t^B+\|\ket{\psi_{t, B}^0}\|^2.
\end{align*}
The last inequality holds from maximizing the quadratic expression over $\|\ket{\psi_{t}^B}\|$.

Then, by the definition of $H_t$,
\begin{align*}
    H_{t+1}-H_{t} &= \sum_{B: f((0^n)^B)=1} \pi_B (H^{B}_{t+1}-H^B_t) \\
    &\leq  \sum_{B: f((0^n)^B)=1} \pi_B\|\ket{\psi_{t, B}^0}\|^2 \\
    &=  \sum_{B: f((0^n)^B)=1} \pi_B \bigg(\sum_{i\,:\,i \in B}\|\alpha_{t, i}^0\|^2\bigg) \\
    &= \frac{1}{W}\sum_i \|\alpha_{t, i}^0\|^2\sum_{B:i \in B}w_B \\  
    &\leq \frac{1}{W}.
    \end{align*}
 The final inequality follows from the fact that $\sum_i \|\alpha_{t, i}^0\|^2 \leq 1$ and that $\sum_{B:i \in B}w_B \leq 1$. 
\end{proof}


\begin{lemma}[One-step growth: quantum]
\label{lem:hybridq}
    For all $B\text{ such that } f((0^n)^B)=1$ and $t\in T_Q$ we have \[H_{t+1}^B - H_{t}^B \leq 4\|\ket{\psi_{t}^B}\|\|\ket{\psi_{t,B}^0}\| \,\,\,\,\text{ and }\,\,\,\, H_{t+1} - H_{t} \leq \frac{4}{\sqrt{W}}.\]
\end{lemma}
\begin{proof} By the definition of $H_{t+1}^B$ we have,
\begin{align*}
H_{t+1}^B &= \left\|\ket{\phi_{t+1}^0}-\ket{\phi_{t+1}^B}\right\|^2 \\
&= \left\|U_{t+1}\ket{\widetilde{\phi}_{t+1}^0}-U_{t+1}\ket{\widetilde{\phi}_{t+1}^B}\right\|^2 \\ 
&= \left\|O_0\ket{\widetilde{\phi}_{t+1}^0}-O_B\ket{\widetilde{\phi}_{t+1}^B}\right\|^2 \\
&= \left\|\ket{\phi_t^0}-\ket{\phi_t^B} +  2\ket{\psi_t^B}\right\|^2 \\
&= H_t+4\|\ket{\psi_{t}^B}\|^2+2\langle{\phi_t^0-\phi_t^B}\,|\,{\psi_{t}^B}\rangle+2\langle{\psi_{t}^B}\,|\,{\phi_t^0-\phi_t^B}\rangle \\
&\leq H_t+4\|\ket{\psi_{t}^B}\|\|\ket{\psi_{t,B}^0}\|.
\end{align*}
By the definition of $H_t$ we have,
\begin{align*}
    H_{t+1}-H_{t} &= \sum_{B: f((0^n)^B)=1} \pi_B (H^{B}_{t+1}-H^B_t) \\
    &\leq  4\sum_{B: f((0^n)^B)=1} \pi_B\left\|\ket{\psi^{0}_{t,B}}\right\|\left\|\ket{\psi^{B}_{t}}\right\| \\
    &= 4\sum_{B: f((0^n)^B)=1} \sqrt{\pi_B^2\left\|\ket{\psi^{0}_{t,B}}\right\|^2\left\|\ket{\psi^{B}_{t}}\right\|^2} \\
    &\leq 4\sqrt{\sum_{B: f((0^n)^B)=1}\pi_B\left\|\ket{\psi^{0}_{t,B}}\right\|^2}\sqrt{\sum_{y: f(y(0^n)^B)=1}\pi_B\left\|\ket{\psi^{B}_{t}}\right\|^2} \\
    &\leq 4\sqrt{\sum_{B: f((0^n)^B)=1}\pi_B\left\|\ket{\psi^{B}_{t}}\right\|^2} \\
    &\leq 4\sqrt{\sum_{B: f((0^n)^B)=1}\pi_B} \\
    &= \frac{4}{\sqrt{W}}.
\end{align*}
The penultimate inequality follows from Lemma~\ref{lem:hybridp}.
The last inequality follows from the fact that $\left\|\ket{\psi^{B}_{t}}\right\|^2 \leq 1$.
The final equality follows from the linear program in Eq.~\eqref{ref:lp} and the definition of fractional block sensitivity.
\end{proof}

Next, we combine Lemmas~\ref{lem:progresslb},~\ref{lem:hybridp} and~\ref{lem:hybridq} to prove a version of Theorem~\ref{thm:hybrid} that considers the fractional block sensitivity at $0^n$.

\begin{theorem}
Any bounded-error algorithm computing $f$ with probability at least $1-\epsilon$ using $C$ pseudo-classical queries and $Q$ coherent quantum
queries satisfies
\[
  C+Q\sqrt{\fbs(f, 0^n)}=\Omega_\eps(\fbs(f, 0^n)).
\]
\label{thm:hybrid:helper}
\end{theorem}

\begin{proof}
 For all $t\in T_c$, by Lemma~\ref{lem:hybridp} we have
 \[H_{C} \leq \sum_{t\in T_c}\frac{1}{W} \, = \, \frac{C}{W}\,.\]
 Similarly, for all $t\in T_q$, by Lemma~\ref{lem:hybridq} we have
 \[H_{Q} \leq \sum_{t\in T_q}\frac{4}{\sqrt{W}} \, = \, \frac{4Q}{\sqrt{W}}\,.\]
 For the hybrid algorithm since $C + Q = T$, for all $t\in T$ we have
 \begin{align*}
     H_T &= \sum_{t\in T_c}H_t + \sum_{t\in T_q}H_t \\
     &\leq H_{C}+H_{Q} \\
     &\leq \frac{C}{W} + \frac{4Q}{\sqrt{W}}\,.
 \end{align*}
 Using Lemma~\ref{lem:progresslb} and multiplying $W = \fbs(f,0^n)$ we have,
 \[C + 4\sqrt{\mathsf{fbs}(f, 0^n)}\,Q \geq \mathsf{fbs}(f, 0^n)(1-4\sqrt{\epsilon}).\]
 \end{proof}

\begin{proof}[Proof of Theorem~\ref{thm:hybrid}]

First due to translation covariance (\cref{lem:hybrid-translation-cq}) for both quantum and classical queries, \cref{thm:hybrid:helper} hence implies that 

\[C+ Q\sqrt{\fbs(f)}=\Omega_\eps(\fbs(f)) \, .\]
We are just left to show that this is equivalent to the expression in~\cref{thm:hybrid} that 

\[ C+Q^2=\Omega_\varepsilon(\fbs(f)) \,. \]

Suppose \( C+Q\sqrt{\fbs(f)} \ge c_\varepsilon \fbs(f) \), then 

\[ \text{ either } \quad  C\ge \frac{c_\varepsilon}{2}\fbs(f); \quad \text{ or } \quad Q\sqrt{\fbs(f)}\ge \frac{c_\varepsilon}{2}\fbs(f)\Rightarrow Q^2\ge \frac{c_\varepsilon^2}{4}\fbs(f) \, . \]

Therefore
\[
C+Q^2
\ge
\min\left\{\frac{c_\varepsilon}{2},
\frac{c_\varepsilon^2}{4}\right\}\fbs(f)
=
\Omega_\varepsilon(\fbs(f)) \ , .\]

The other direction is similar.  Suppose \( C+Q^2\ge c_\varepsilon \fbs(f) \). 
\begin{itemize}

\item  If \(Q\le\sqrt{\fbs(f)}\), then \(C+Q\sqrt{\fbs(f)}
\ge C+Q^2
\ge c_\varepsilon \fbs(f)\). 
\item If \(Q>\sqrt W\), then \(Q\sqrt{\fbs(f)}>\fbs(f)\), and \( C+Q\sqrt{\fbs(f)}\ge \fbs(f)
\ge \min\{1,c_\varepsilon\} \fbs(f)\).
\end{itemize}
Hence \(
C + Q\sqrt{\fbs(f)}=\Omega_\varepsilon \fbs(f) \). 
\end{proof}

\section{IID Query-Index Dephasing}
\label{sec:iid}
In this section, we prove~\cref{thm:iid-intro}. Due to translation covariance~\cref{lem:deph-translation}, it suffices to prove it with respect to \(\fbs\) at \(0^n\). 

\begin{theorem}[IID dephasing bound]
\label{thm:iid-anchor}
For a function \(f\) and \(0<\eps <1/2\), iid query-index dephasing at rate $p$ obeys
\[
  Q^{\deph}_{p,\eps}(f) =\Omega_\eps\!\left(p\fbs(f,0^n)\right).
\]
\end{theorem}

\subsection{Coherent weighted-star potential analysis}
\label{sec:local-potential}

As before, fix a fractional sensitive-block packing $(w_B)$ at the center $0^n$, and denote
\[
  W=\sum_Bw_B,
  \qquad
  \pi_B=\frac{w_B}{W},
  \qquad
  \sum_B\pi_B=1.
\]

\begin{figure}[t]
\centering
\begin{tikzpicture}[
    x=1.05cm,
    y=0.68cm,
    >=latex,
    thick,
    every node/.style={font=\small}
]

\def\yR{5}
\def\yT{4}
\def\yN{3}
\def\yA{2}
\def\yY{1}
\def\yW{0}

\node[anchor=east] at (0,\yR) {$R$};
\node[anchor=east] at (0,\yT) {$T$};
\node[anchor=east] at (0,\yN) {$N_t$};
\node[anchor=east] at (0,\yA) {$I$};
\node[anchor=east] at (0,\yY) {$O$};
\node[anchor=east] at (0,\yW) {$W$};

\node[anchor=west] at (0.15,\yR+0.25) {\scriptsize $|r_{t-1}\rangle$};
\node[anchor=west] at (0.15,\yT+0.25) {\scriptsize $|u\rangle$};
\node[anchor=west] at (0.15,\yN+0.25)
    {\scriptsize $\sqrt{1-p}\,|0\rangle+\sqrt p\,|1\rangle$};

\foreach \y in {\yR,\yT,\yN,\yA,\yY,\yW}
    \draw (0.15,\y) -- (10.7,\y);

\draw[rounded corners=2pt, fill=white]
    (1.8,-0.35) rectangle (3.15,2.35);
\node at (2.475,1.0) {$U_t$};

\fill (4.45,\yT) circle (2.2pt);
\draw (4.45,\yT) -- (4.45,2.35);

\draw[rounded corners=2pt, fill=white]
    (3.85,0.65) rectangle (5.05,2.35);
\node at (4.45,1.50) {$O_T$};

\draw[dashed, rounded corners=3pt]
    (6.0,-0.55) rectangle (9.15,5.45);

\node at (7.575,5.75)
    {\scriptsize purified dephasing};

\fill (6.85,\yN) circle (2.2pt);

\draw (6.85,\yA) -- (6.85,\yR);
\fill (6.85,\yA) circle (2.2pt);

\draw[rounded corners=2pt, fill=white]
    (6.35,4.68) rectangle (7.35,5.32);
\node at (6.85,\yR) {\scriptsize rec};

\fill (8.25,\yN) circle (2.2pt);
\draw (8.25,\yN) -- (8.25,\yW);

\draw[rounded corners=2pt, fill=white]
    (7.75,-0.32) rectangle (8.75,0.32);
\node at (8.25,\yW) {\scriptsize flag};


\node[anchor=west] at (9.35,\yR+0.35)
    {\scriptsize hidden record};
\node[anchor=west] at (9.35,\yW + 0.35)
    {\scriptsize workspace + flags};

\end{tikzpicture}

\caption{
One step of the purified IID dephasing computation.
The truth register is initialized in the weighted-star state
\(
|u\rangle = \frac1{\sqrt2}
  \left(
    |0\rangle+\sum_B\sqrt{\pi_B}\,|B\rangle
  \right).
\)
After the algorithm unitary $U_t$ and the truth-controlled oracle $O_T$,
a fresh noise qubit coherently selects the undephased and dephased
branches.  On the dephased branch, the queried address is appended to
the inaccessible record $R$; the corresponding dephasing flag is stored
in the algorithm workspace $W$.
}
\label{fig:purified-iid-query}
\end{figure}
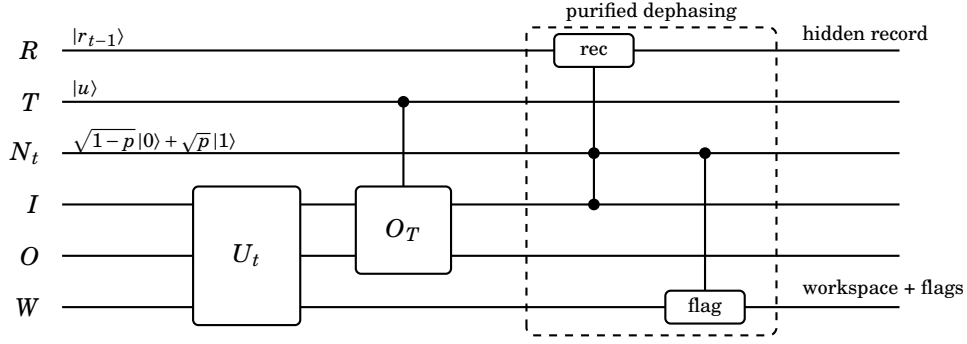

Figure~\ref{fig:purified-iid-query} illustrates one step of computation in our analysis. We initialize a truth table register in 
\[
  \ket{u} = \frac1{\sqrt2}
  \left(
    |0\rangle+\sum_B\sqrt{\pi_B}\,|B\rangle
  \right).
\]

Purify every dephasing by copying its query index into an inaccessible record.  For an exact record string $r$, let $R=R(r)$ be the set of indices that occur
in dephased slots and define
\[
  \cB_R=\{B:B\cap R=\varnothing\},
  \qquad
  b_R=\sum_{B\in\cB_R}\pi_B.
\]
In the truth register, set
\[
  |u_R\rangle
  =\frac{|0\rangle+\sum_{B\in\cB_R}\sqrt{\pi_B}\,|B\rangle}{\sqrt{1+b_R}},
  \qquad
  \cH_{A,R}=\operatorname{span}\{|u_R\rangle\},
\]
\[
  \cH_{C,R}=\operatorname{span}\{|B\rangle:B\notin \cB_R\},
  \qquad
  \cH_{Q,R}=(\cH_{A,R}\oplus\cH_{C,R})^\perp.
\]
The global $A/Q/C$ spaces are the orthogonal direct sums of these spaces over
the hidden records, tensored with all algorithm registers.  These subspaces are
defined by the original packing and the hidden record. Projectors onto these subspaces can be defined as 
\[
    \PiA{R} = \ketbra{u_R},
    \qquad
    \PiC{R} = \sum_{B\notin \cB_R}\ketbra{B},
    \qquad
    \PiQ{R} = I - (\PiA{R}+\PiC{R})\,.
\]

Let
\[
  q_t=\|\Pi_Q\Omega_t\|^2,
  \qquad
  c_t=\|\Pi_C\Omega_t\|^2,
  \qquad
  \Psi_t=c_t+3q_t,
\]
where $|\Omega_t\rangle$ denotes the normalized conditional truth--algorithm
state on a positive-probability accessible history\footnote{\(h_{t-1} = (Z_1, \ldots, Z_{t-1}, \text{ measurement outcomes, revealed coins/control data }, \ldots ) \), which is  a particular realized classical transcript up to that time.}. Initially the state lies
in $A$, so $\Psi_0=0$.

\paragraph{The active two-dimensional truth plane.}
Fix a record $R$ and a query index $i$, and define the normalized hit direction
\[
    \ket{h_{R,i}}
  :=
  \frac{1}{\sqrt{k_{R,i}}}
  \sum_{\substack{B\in\cB_R\\i\in B}}
       \sqrt{\pi_B}\,|B\rangle \, , \quad k_{R,i}=\sum_{\substack{B\in\cB_R\\i\in B}}\pi_B \, ,
\]
assuming $k_{R,i}>0$. Its overlap with the unresolved direction is
\[
  \langle u_R,h_{R,i}\rangle
  =
  \sqrt{\frac{k_{R,i}}{1+b_R}}
  =:\alpha_{R,i}.
\]
Put
\[
  \beta_{R,i}:=\sqrt{1-\alpha_{R,i}^2}.
\]
There is therefore a unit vector $\ket{e_{R,i}}\in\cH_{Q,R}$ such that
\begin{equation}
  \ket{h_{R,i}}
  =
  \alpha_{R,i}\ket{u_R}+\beta_{R,i}\ket{e_{R,i}}.
  \label{eq:pht-hit-decomposition}
\end{equation}
The vector $\ket{e_{R,i}}$ is the unique $Q$-direction that can exchange amplitude with the unresolved \(A\)-direction upon a query to index \(i\); we call it the \emph{active $Q$-direction}.
All other $Q$-directions are passive for this particular query. For a fixed $R$ and $i$, Projectors onto the active-Q-direction passive-Q-directions are
\[
\Pi^{Q,Act}_{R,i} = \ketbra{e_{R,i}}, 
\qquad
\Pi^{Q,Pas}_{R,i} = \Pi^Q_R - \Pi^{Q,Act}_{R,i}\,.
\]

The orthogonal no-hit direction is
\begin{equation}
  \ket{\ell_{R,i}}
  :=
  \beta_{R,i}\ket{u_R}-\alpha_{R,i}\ket{e_{R,i}}
  =
  \frac{
    |0\rangle+
    \sum_{\substack{B\in\cB_R\\i\notin B}}
      \sqrt{\pi_B}\,|B\rangle
  }{
    \sqrt{1+b_R-k_{R,i}}
  }.
  \label{eq:pht-nohit-direction}
\end{equation}
Hence
\[
  \operatorname{span}\{\ket{u_R},\ket{e_{R,i}}\}
  =
  \operatorname{span}\{\ket{\ell_{R,i}},\ket{h_{R,i}}\}.
\]

\begin{figure}[ht]
\centering
\begin{tikzpicture}[scale=2.1,>=Latex]
  \coordinate (O) at (0,0);

  \draw[->] (O) -- (2.15,0)
    node[above] {$\ket{u_R}\in {A_R}$};
  \draw[->] (O) -- (0,1.8)
    node[above] {$\ket{e_{R,i}}\in {Q_R}$};
  \def\hcoor{3.7}
  \def\ephi{18}

  \draw[->,thick]
    (O) -- ({1.75*cos(\ephi)},{-1.75*sin(\ephi)})
    node[right] {$\ket{\ell_{R,i}}$};

  \draw[->,thick]
    (O) -- ({1.55*sin(\ephi)},{1.55*cos(\ephi)})
    node[above right] {$\ket{h_{R,i}}$};

  \draw[dashed]
    (1.15,0) arc[start angle=0,end angle=-\ephi,radius=1.15];

  \node at (1.05,-0.18) {$\phi$};

  \node[align=left] at (\hcoor,1.35) {
    $\alpha=\sin\phi,\quad\beta=\cos\phi$\\[1mm]
    $h=\alpha u+\beta e$\\
    $\ell=\beta u-\alpha e$
  };

  \node[align=left] at (\hcoor,0.5) {
    coherent query:\\
    $D_i\ell=\ell,\quad D_i k=-k$
  };

  \node[align=left] at (\hcoor,-0.3) {
    after dephasing:\\
    $\ell\rightarrow A_{R\cup\{i\}}$\\
    $h\rightarrow C_{R\cup\{i\}}$
  };

\end{tikzpicture}
\caption{The active two-dimensional truth plane for record $R$ and queried
index $i$.  The basis $(\ket{u_R},\ket{e_{R,i}})$ is adapted to the old $A/Q$
decomposition, while $(\ket{\ell_{R,i}},\ket{h_{R,i}})$ is adapted to the no-hit/hit
decomposition of the query.}
\label{fig:pht-active-plane}
\end{figure}
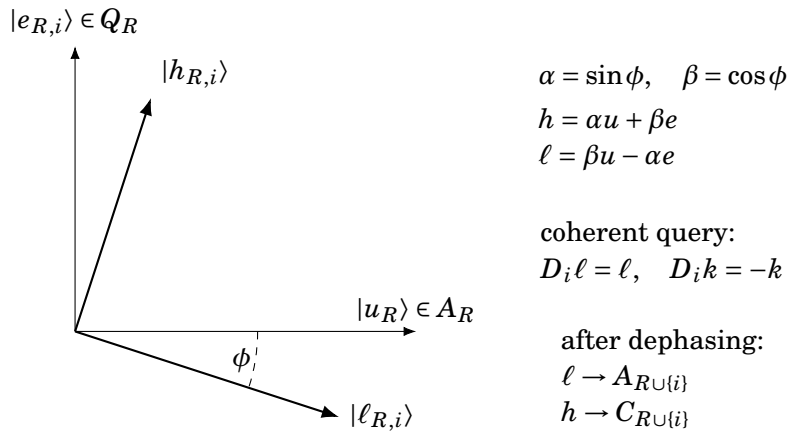

Observe that a coherent quantum query
acts as a reflection in this two-dimensional plane about \(\ket{\ell_{R,i}}\), and we call it \(D_i\): 
\[
  D_i\ket{\ell_{R,i}}=\ket{\ell_{R,i}},
  \qquad
  D_i \ket{h_{R,i}}=-\ket{h_{R,i}}\, .
\]

On the other hand, if the query is subsequently dephased and the record is extended by $i$,
then
\[
  \ket{\ell_{R,i}}\in\cH_{A,R\cup\{i\}},
  \qquad
  \ket{h_{R,i}}\in\cH_{C,R\cup\{i\}}.
\]
As we will see, this is the source of the negative quadratic term in the
dephased-query potential bound. 

Let $a_t$ be the norm of the full active-$Q$ component, summed over all orthogonal record, query-index, answer, and workspace components. Namely, let 

\[a_t^2=\|\Pi^{\mathrm{Q, act}}|\Omega_t\rangle\|^2 \, \]
where 

\[
\Pi^{\mathrm{Q,act}}
=
\sum_{R,i}
|R\rangle\langle R|
\otimes
|i\rangle\langle i|_I
\otimes
|e_{R,i}\rangle\langle e_{R,i}|
\otimes I,
\]
i.e. \(a^2\) is the active \(B\)-mass, while \(a\) is its magnitude.  Then \(0\le a_t^2\le q_{t-1}\le1 \). 
  
\begin{lemma}[One-query growth bound]
\label{lem:local-potential}
Conditioned on the accessible history immediately before query~$t$,
\[
  \E_t[\Psi_t-\Psi_{t-1}]
  \le
  \frac{12}{W}
  +
  \frac{12a_t}{\sqrt W}
  -2p a_t^2.
\]
Here \(\E_t\) denotes expectation over the fresh independent dephasing choice at query \(t\), conditional on everything that has happened before that query.
\end{lemma}

\begin{proof}
Fix a positive-probability history and its normalized conditional state.  \paragraph{Coherent query.}
Consider one orthogonal record/index/workspace component and suppress the
subscripts $R,i$.  Write it as
\[
  \ket{u}\otimes \ket{x}+\ket{e}\otimes \ket{y},
\]
where $\ket{u}$ is the unresolved $A$-direction, $\ket{e}$ is the active $Q$-direction,
and $\ket{x},\ket{y}$ are vector-valued amplitudes in the remaining registers.

In the basis $(\ket{u},\ket{e})$, the phase form of the bit query is the reflection
\begin{equation}
  D_i
  =
  \begin{pmatrix}
    1-2\alpha^2 & -2\alpha\beta\\
    -2\alpha\beta & 2\alpha^2-1
  \end{pmatrix},
  \qquad
  \beta=\sqrt{1-\alpha^2}.
  \label{eq:pht-query-matrix}
\end{equation}

Thus the new active-$Q$ amplitude is
\(\ket{y'} = -2\alpha\beta \ket{x}+(2\alpha^2-1)\ket{y}\), which gives
\[
  \|y'\|^2 =  4\alpha^2\beta^2\|x\|^2
  +(2\alpha^2 -1)^2\|y\|^2 -
  4\alpha\beta(2\alpha^2-1)
       \operatorname{Re}\langle x,y\rangle.
\]
Using $\beta\le1$ and $|2\alpha^2-1|\le1$ therefore gives
\[
  \|y'\|^2-\|y\|^2
  \le
  4\alpha^2\|x\|^2
  +4\alpha\|x\|\,\|y\|.
\]
The $B$-subspace has weight three in $\Psi=c+3q$, and the coherent query
does not change the hidden record.  Hence this record–query component contributes at most
\[
  12\alpha^2\|x\|^2
  +12\alpha\|x\|\,\|y\|.
\]

Summing over the mutually orthogonal active record–query component, using
$\alpha^2\le1/W$, and noting
\(a_t^2=\sum_{r,i}\|y\|^2\), Cauchy--Schwarz gives
\[
\begin{aligned}
  \Delta\Psi_{\rm coh}
  &\le
  12\sum_{r,i}\alpha^2\|x\|^2
  +
  12\sum_{r,i}\alpha\|x\|\,\|y\|\\
  &\le
  \frac{12}{W}
  +
  \frac{12a_t}{\sqrt W}.
\end{aligned}
\]

\paragraph{Dephased query.}
For the same active plane,
\[
  \ket{u}=\beta\ket{\ell}+\alpha \ket{h},
  \qquad
  \ket{e}=-\alpha\ket{\ell}+\beta \ket{h},
\]
and therefore
\[
  \ket{u}\otimes \ket{x}+\ket{e}\otimes \ket{y}
  =
  \ket{\ell}\otimes(\beta \ket{x}-\alpha \ket{y})
  +
  \ket{h}\otimes(\alpha \ket{x}+\beta \ket{y}).
\]
The phase query fixes $\ket{\ell}$ and flips the sign of $\ket{h}$.  After the subsequent
query-index dephasing, the hidden record is extended by $i$.  Relative to the
new record,
\[
  \ket{\ell}\in\cH_{A,R\cup\{i\}},
  \qquad
  \ket{h}\in\cH_{C,R\cup\{i\}}.
\]
Thus the active plane contains no $Q$-component after the dephased query.
Its new contribution to the potential  from the $\ket{h}$ component,
\[
  \|\alpha \ket{x}+\beta \ket{y}\|^2,
\]
whereas before the query its active $Q$-component contributed
$3\|y\|^2$.  Hence
\[
\begin{aligned}
  \Delta\Psi_{\rm deph,active}
  &=
  \|\alpha \ket{x}+\beta \ket{y}\|^2-3\|y\|^2\\
  &=
  \alpha^2\|x\|^2
  -(2+\alpha^2)\|y\|^2
  +2\alpha\beta\operatorname{Re}\langle x,y\rangle\\
  &\le
  \alpha^2\|x\|^2
  +2\alpha\|x\|\,\|y\|
  -2\|y\|^2.
\end{aligned}
\]

Passive no-hit $Q$-mass remains in $Q$.  Passive hit $Q$-mass moves from
$Q$ to $C$, changing its potential coefficient from three to one, and
therefore contributes nonpositively.  Summing the active estimate over all
orthogonal blocks gives
\[
  \Delta\Psi_{\rm deph}
  \le
  \frac1W+\frac{2a_t}{\sqrt W}-2a_t^2.
\]

Hence the law of total expectation gives
\begin{align*}
  \E[\Psi_t-\Psi_{t-1}\mid\cF_{t-1}]
  =& 
  (1-p)\cdot
  \E[\Psi_t-\Psi_{t-1}\mid\cF_{t-1},Z_t=0]
  + 
  p \cdot
  \E[\Psi_t-\Psi_{t-1}\mid\cF_{t-1},Z_t=1]\\
  =& (1-p_t)
  \left(
    \frac{12}{W}+\frac{12a_t}{\sqrt W}
  \right) + p_t \left(
    \frac1W+\frac{2a_t}{\sqrt W}-2a_t^2
  \right)\\
  =&  \frac{12-11p_t}{W}
  +
  \frac{(12-10p_t)a_t}{\sqrt W}
  -
  2p_ta_t^2\\
  \le&
  \frac{12}{W}
  +
  \frac{12a_t}{\sqrt W}
  -
  2p_ta_t^2.
\end{align*}
\end{proof}

\begin{lemma}[Record-dependent terminal progress]
\label{lem:terminal}
If the algorithm has worst-case error at most $\eps<1/2$, then
\[
  \E\Psi_T\ge\gamma_\eps,
  \qquad
  \gamma_\eps
  :=\frac{\bigl(1-2\sqrt{\eps(1-\eps)}\bigr)^2}{36}>0.
\]
Expectation here is over the IID dephasing choices and any other classical transcript branches generated by the purified execution. 
\end{lemma}

\begin{proof}
We work with the unconditional common purification, obtained by retaining all
measurement, randomness, and transcript registers.  Thus taking the expectation of
the history-conditioned subspace masses gives the final masses in each subspace. Write
\[
  \bar q_T:=\E q_T,
  \qquad
  \bar c_T:=\E c_T,
  \qquad
  s:=\bar q_T+\bar c_T.
\]
Write the final common-purification state as
\begin{equation}
  |\Omega_T\rangle
  =\frac1{\sqrt2}
  \left(
    |0\rangle\sum_r|r\rangle|x_r\rangle
    +\sum_B\sqrt{\pi_B}|B\rangle
       \sum_r|r\rangle|y_{B,r}\rangle
  \right).
  \label{eq:pht-terminal-state}
\end{equation}
Here $r$ ranges over the inaccessible dephasing records, while all remaining
purification, accessible-transcript, and algorithm registers are included in
$x_r$ and $y_{B,r}$.  Denote 
\[ |\psi_0\rangle =\sum_r\ket{r}\ket{x_r}, \quad |\psi_B\rangle = \sum_r \ket{r}\ket{y_{B,r}} \, , \]
and define 
\[
  H: =\sum_B\pi_B\|\psi_0-\psi_B\|^2 \, .
\]

Assuming the algorithm succeeds with error \(\eps\), then the final measurement distinguishes the center from every leaf with
error at most $\eps$.  Monotonicity of fidelity under this measurement gives
\[
  |\langle\psi_0|\psi_B\rangle|
  \le2\sqrt{\eps(1-\eps)} \text{ for all } B\, .
\]

Consequently
\[
  H\ge2\bigl(1-2\sqrt{\eps(1-\eps)}\bigr)\, .
\]

Next we show that 
\begin{equation}    
  H\le12\sqrt s \, ,
  \label{eq:pht-terminal-D-s}
\end{equation}

and the lemma follows immediately
\[
  \E\Psi_T
  =\bar c_T+3\bar q_T \ge s\ge
  \frac{\bigl(1-2\sqrt{\eps(1-\eps)}\bigr)^2}{36} \, .
\]

To show~\cref{eq:pht-terminal-D-s}, since different \(r\) components are orthogonal, we have that
\begin{align*}
    H & = \sum _B \pi_B \sum_r  \| x_r - y_{B,r} \|^2 \\
      & = \sum_r \sum _{B\in\cB_R} \pi_B  \| x_r - y_{B,r} \|^2 + \sum_r \sum _{B\notin\cB_R} \pi_B  \| x_r - y_{B,r} \|^2 \\
       & \le \sum_r \sum _{B\in\cB_R} \pi_B  \| x_r - y_{B,r} \|^2  + \sum_r \sum _{B\notin\cB_R} 2\pi_B  (\| x_r\|^2 + \| y_{B,r} \|^2) \\
      & \le d+2X+2Y \, ,
\end{align*}  
by denoting 
\[
  H=\sum_r\sum_{B\in\cB_R}
       \pi_B\|x_r-y_{B,r}\|^2,
  \qquad
  Y =\sum_r\sum_{B\notin\cB_R}
       \pi_B\|y_{B,r}\|^2,
  \label{eq:pht-terminal-dy}
\]
\[
  X=\sum_r(1-b_r)\|x_r\|^2,
  \qquad
  b_r=\sum_{B\in\cB_R}\pi_B.
  \label{eq:pht-terminal-X}
\]

Next we bound \(d,y,X\) respectively. For a fixed \(r\), let \(z_r\) be the  vector on the algorithm registers that describes the effective projection onto the \(A\)-direction. Then the subspace \(Q\)-component is

\[
\frac1{\sqrt2}
\left[
|0\rangle(x_r-z_r)
+
\sum_{B\in\cB_R}
\sqrt{\pi_B}|B\rangle(y_{B,r}-z_r)
\right].
\]

and denote its squared norm 
\[
q_r
:=
\frac12\left(
\|x_r-z_r\|^2+
\sum_{B\in\cB_R}
\pi_B\|y_{B,r}-z_r\|^2
\right).
\]

Now for each unresolved \(B\),
\[ 
\|x_r-y_{B,r}\|^2
\le
2\|x_r-z_r\|^2
+
2\|y_{B,r}-z_r\|^2.
\]

Multiply by \(\pi_B\) and sum:
\begin{align*}
d_r
&:=
\sum_{B\in\cB_R}
\pi_B\|x_r-y_{B,r}\|^2\\
&\le
2b_R\|x_r-z_r\|^2
+
2\sum_{B\in\cB_R}
\pi_B\|y_{B,r}-z_r\|^2 \\
& \le 2\left(
\|x_r-z_r\|^2+
\sum_{B\in\cB_R}
\pi_B\|y_{B,r}-z_r\|^2
\right) \\
& = 4q_r \, ,
\end{align*}
since \(b_R\le1\). Finally sum over records \(r\):
\[
d=\sum_r d_r\le4\sum_r q_r=4\bar q_T \le 4 s \, . \]

To bound $Y$, note that the terms all belong to the \(C\)-subspace, hence 
\[ Y \le 2 \bar c_T \le 2s \, . \]

Finally moving to \(X\), consider 
\[
U :=\bigoplus_{\substack{r,B\\ B\cap R(r)}}
\sqrt{\pi_B}\,y_{B,r}\, , \quad
V := \bigoplus_{\substack{r,B\\ B\cap R(r)}} \sqrt{\pi_B}\,x_r \, .
\]

Then
\[
  \|U\|^2=1-Y,
  \qquad
  \|V\|^2=1-X,
  \qquad
  \|U-V\|^2=d.
\]
The triangle inequality gives
\[
  X = 1 - \|U - (U-V)\|^2 \le1-\bigl(\sqrt{1-Y}-\sqrt d\bigr)^2 = Y- d +  2\sqrt{(1-Y)d} \, .
\]
Hence 
\[
  H \le d+2X+2Y \le-d+4Y+4\sqrt{d(1-Y)}
  \le8s+8\sqrt s<12\sqrt s \, ,
\]
whenever \(s<1/9\); and if $s\ge1/9$, \(H\le 12 \sqrt s \) follows from $H\le4$. 
\end{proof}

\subsection{Putting It Together}
\label{sec:iid-fbs-0}
\begin{proof}[Proof of Theorem~\ref{thm:iid-anchor}]
By Lemma~\ref{lem:local-potential}, for every positive-probability history,
\[
    \E_t\!\left[
        \Psi_t-\Psi_{t-1}
    \right]
    \le
    \frac{12}{W}
    +\frac{12a_t}{\sqrt W}
    -2p_ta_t^2.
\]    

Taking expectations and summing over \(t=1,\ldots,T\) gives

\[  \E[\Psi_T]
    \le
    \sum_{t=1}^T
    \E\!\left[
        \frac{12}{W}
        +\frac{12a_t}{\sqrt W}
        -2pa_t^2
    \right].
\]

Note that
\[
    \frac{12a_t}{\sqrt W}-2pa_t^2
    =
    \frac{18}{pW}
    -
    \left(
        \sqrt{2p}\,a_t
        -
        \frac{6}{\sqrt{2pW}}
    \right)^2
    \le
    \frac{18}{pW}.
\]

Hence

\[
    \E[\Psi_T]
    \le
    T\left(
        \frac{12}{W}
        +\frac{18}{pW}
    \right)
    =
    \frac{18+12p}{pW}\,T.
\]

Lemma~\ref{lem:terminal} asserts the success terminal condition \(
    \E[\Psi_T]\ge \gamma_\eps \), and thus

\[
    T
    \ge
    \frac{\gamma_\eps\,pW}{18+12p} = \Omega_\eps(p \fbs(f,0^n)) \, .
\]

\end{proof}

\section{An Extended Family of Stochastic Dephasing}
\label{sec:ext-dephase}
\subsection{Algorithm-independent stochastic dephasing}
\label{subsec:alg-indep-dephasing}

Fix a query number $T$.  A dephasing schedule is a binary vector
\[
  Z=(Z_1,\ldots,Z_T)\in\{0,1\}^T,
\]
where $Z_t=1$ means that $t$th query is dephased.  We allow an arbitrary joint
law for this vector, but require it to be sampled independently of the hidden
input, the algorithm, its private randomness, and all computational registers.
Equivalently, the entire schedule may be sampled before the computation begins
and then revealed one coordinate at a time.

\begin{definition}[Algorithm-independent dephasing schedule]
\label{def:dephasing-schedule}
For a fixed query count $T$, define the set of all distributions on \(\bit^T\)
\[
  \mathcal D(T)
  :=
  \left\{
    \mu:\{0,1\}^T\to[0,1]
    \;\middle|\;
    \sum_{z\in\{0,1\}^T}\mu(z)=1
  \right\}.
\]
An algorithm-independent dephasing schedule is a random vector
\[
  Z=(Z_1,\ldots,Z_T)\sim\mu,
  \qquad
  \mu\in\mathcal D(T) \, .
\]
\end{definition}

Note that \(Z\) is sampled independently of the hidden input, the algorithm, its private
randomness, and all computational registers. Otherwise, no independence or Markov assumption is imposed on the coordinates of $Z$.
Let
\begin{equation}
  \cG_t:=\sigma(Z_1,\ldots,Z_t),
  \qquad
  p_t:=\Pr_\mu[Z_t=1\mid\cG_{t-1}].
  \label{eq:conditional-dephasing-rate}
\end{equation}
Here \(\cG_t\) represents the information generated by the revealed dephasing flags up to time $t$. Similarly let $\cF_t$ denote the full accessible classical information available to the
algorithm after query $t$, which includes the flags $Z_1,\ldots,Z_t$, all measurement outcomes,
revealed random coins, and classical control variables generated so far. It is clear that \(\cG_t\subseteq\cF_t \). We hence say that the conditional rate $p_t$ is
\emph{predictable}\footnote{In measure-theory language, each \(\cG_t\) or \(\cF_t\) is a \(\sigma\)-algebra representing the information available by time \(t\), and Then \(p_t\) is \(\cG_{t-1}\)-measurable and hence also \(\cF_{t-1}\)-measurable. The sequences
\(
(\cG_t)_{t=0}^T, (\cF_t)_{t=0}^T
\)
are filtrations because they are increasing. We may call \((\cG_t)\) the noise filtration and \((\cF_t)\) the full accessible filtration. .}.

Conditional on the
realized value of $Z_t$, the query channel is
\[
  \rho\longmapsto
  \begin{cases}
    O_x\rho O_x^\dagger, & Z_t=0,\\[1mm]
    \Delta_I(O_x\rho O_x^\dagger), & Z_t=1,
  \end{cases}
  \qquad
  \Delta_I(\rho)=\sum_iP_i\rho P_i\, , \quad P_i = |i\rangle \langle i | \otimes I \, .
\]
The flag $Z_t$ is revealed only after the channel has acted.

This model includes deterministic, IID, Markov, and
arbitrarily long-range correlated schedules.  It excludes feedback from the
algorithm to the noise: future dephasing cannot depend on the hidden input,
measurement outcomes, control choices, current query amplitudes, or hidden
query-index records created by previous dephasings.

\subsection{Computational-opportunity resource and fbs lower bound}
\label{subsec:computational-opportunity}

We consider a \emph{predictable opportunity-allocation process} 
$\theta=(\theta_t)_{t=1}^T$ such that 
\[
  \theta_t=\vartheta_t(Z_1,\ldots,Z_{t-1})\, ,
\]
for some fixed function $\vartheta_t$.  It is
called a process because its value is evaluated on the random past schedule,
and predictable because it cannot inspect $Z_t$ or future flags\footnote{In other words, $\theta_t\in[0,1]$ and
$\cG_{t-1}$-measurable.}.

\begin{definition}[Opportunity resource]
\label{def:resource}
For $\mu\in\mathcal{D}(T)$ and a predictable
opportunity-allocation process $\theta$, define
\[
  r_\theta(\mu)
  :=\E_\mu\!\left[
      \left(\sum_{t=1}^T\theta_t\right)^2
      +\sum_{t=1}^T\frac{(1-\theta_t)^2}{p_t}
    \right],
\qquad
  \cR(\mu):=\inf_\theta (r_\theta(\mu)) \, .
\]
We use the extended-real convention
\[
  \frac{(1-\theta_t)^2}{p_t}
  =
  \begin{cases}
    0, & p_t=0\text{ and }\theta_t=1,\\
    +\infty, & p_t=0\text{ and }\theta_t\ne1.
  \end{cases}
\]
\end{definition}
Note that finite cost forces $\theta_t=1$ where $p_t=0$. We show that \(\cR(\mu)\) needs to be sufficiently large in order to compute \(f\) under dephasing schedule \(\mu\). 

\begin{theorem}[Computational-Opportunity Theorem]
\label{thm:opportunity}
Let $\mu\in\mathcal{D}(T)$.  If a $T$-query algorithm computes
$f$ with worst-case error $\eps<1/2$ under the dephasing schedule $\mu$, then
\[
  \cR(\mu)\ge c_\eps\fbs(f),
\]
where $c_\eps>0$ depends only on $\eps$ and universal constants.
\end{theorem}

We prove the theorem in~\cref{sec:proof-cot}.  Intuitively, the value $\theta_t$ can be thought of as the share of the one-step progress estimated through
coherent accumulation and such shares add at the amplitude level.  The complementary share $1-\theta_t$ is estimated through
dephasing. The accumulation hence serves as a proxy of the opportunity for an algorithm (not necessarily physically realizable) to make progress. This converts noisy query complexity to an optimization problem oblivious to quantum algorithms, and we showcase a few applications below. 

\subsubsection{Application: hybrid and IID specializations}
\label{sec:special-dephase}

\begin{corollary}[Deterministic $0/1$ schedule]
\label{cor:pht-hybrid}
If the dephasing schedule \(\mu\) is deterministic $\{0,1\}$-valued, with $Q$
quantum coherent and $C$ dephased queries, then \( \cR(\mu) = Q^2+C, Q\ge1\), and hence 
\( C +Q^2=\Omega_\eps(\fbs(f))\). 
\end{corollary}

\begin{proof}
Every coherent query slot has $p_t=0$ and therefore forces $\theta_t=1$.  Let
$x_1,\ldots,x_C$ be the values of $\theta$ on the dephased queries.  The
pathwise objective is
\[
  F(x)=\left(Q+\sum_{j=1}^Cx_j\right)^2
       +\sum_{j=1}^C(1-x_j)^2.
\]
If $Q\ge 1$ and $y=\sum_jx_j$, then
\[
  F(x)-(Q^2+C)
  =2(Q-1)y+y^2+\sum_jx_j^2\ge0,
\]
with equality at $x_j=0$. This proves the corollary. 

For completeness, if $Q=0$, Cauchy--Schwarz reduces the
minimization to \(y^2+{(C-y)^2}/{C}\), 
whose minimizer is $y=C/(C+1)$, equivalently
$x_j=1/(C+1)$.  This gives $C^2/(C+1)$. Therefore, 
\[
  \cR(\mu)
  =
  \begin{cases}
    Q^2+C, & Q\ge1,\\
    C^2/(C+1), & Q=0,\ C\ge1,\\
    0, & Q=C=0.
  \end{cases}
\] 
\end{proof}

\begin{corollary}[IID dephasing]
\label{cor:pht-iid}
If $Z_1,\ldots,Z_T$ are IID Bernoulli variables with dephasing probability
$p$, then
\[
  \cR(\mu)=\frac{T^2}{1+pT}.
\]
Hence
\[
  T=\Omega_\eps\!\left(\sqrt{\fbs(f)}+p\,\fbs(f)\right).
\]
\end{corollary}

\begin{proof}
Let \(S=\sum_t\theta_t \), and
Cauchy--Schwarz gives
\[
  S^2+\frac1p\sum_{t=1}^T(1-\theta_t)^2
  \ge
  S^2+\frac{(T-S)^2}{pT}
  \ge\frac{T^2}{1+pT}.
\]
Equality is attained by the deterministic rule
\[
  \theta_t=\frac1{1+pT}
  \qquad(1\le t\le T) \, .
\]
For $p=0$, finite cost forces $\theta_t=1$, and the same formula gives
$T^2$.  Theorem~\ref{thm:opportunity} now gives
\[
  \frac{T^2}{1+pT}=\Omega_\eps(\fbs(f)) \, ,
\]
which separately implies
$T=\Omega_\eps(\sqrt{\fbs(f)})$ and
$T=\Omega_\eps(p\fbs(f))$. 
\end{proof}

\subsubsection{Application: A Tight Aging-Noise Application}
\label{sec:powerlaw}

Here we apply the theorem to a new model that reflects a possible real-world scenario: as computation progresses, later stages become less reliable because of cumulative or slowly varying imperfections. For cleanliness, we consider 
\[
  p_t=\min\{1,\lambda t^\gamma\},
  \qquad 0<\gamma<1,
  \qquad \lambda>0.
  \label{eq:pht-powerlaw-profile}
\]
where \(\lambda\) controls the overall noise scale and
\(\gamma\) controls the rate at which the noise worsens with time. Then define 
\[ \mu_{\lambda,\gamma}^{(T)}
  =\bigotimes_{t=1}^T\operatorname{Bernoulli}(p_t) \, .\]

\begin{theorem}[Increasing power-law dephasing for search]
\label{thm:powerlaw}
For every fixed $\eps<1/2$ and $0<\gamma<1$, in the flagged model,
\[
Q^{(p_t)}_\eps(\OR_N)
=\Theta_{\eps,\gamma}
\begin{cases}
\sqrt N,
&\lambda\lesssim N^{-(1+\gamma)/2},\\[1mm]
(\lambda N)^{1/(1-\gamma)},
&N^{-(1+\gamma)/2}\lesssim\lambda\lesssim N^{-\gamma},\\[1mm]
N,
&\lambda\gtrsim N^{-\gamma}.
\end{cases}
\]
Boundary conventions affect only constants.  The same upper and lower bounds
hold when the error flags are withheld. 
\end{theorem}

\paragraph{Lower bound from computational opportunity}

For $\OR_N$, the singleton packing at $0^N$ gives $W=N$.  The two
deterministic trial rules $\theta_t\equiv1$ and $\theta_t\equiv0$ yield
\[
  \cR(\mu_{\lambda,\gamma}^{(T)})\le T^2,
  \qquad
  \cR(\mu_{\lambda,\gamma}^{(T)})\le
  A_\lambda(T):=\sum_{t=1}^T\frac1{p_t}.
  \label{eq:pht-powerlaw-trials}
\]
Theorem~\ref{thm:opportunity} therefore forces both
\[
  T^2=\Omega_\eps(N),
  \qquad
  A_\lambda(T)=\Omega_\eps(N).
  \label{eq:pht-powerlaw-necessary}
\]

Put
\[
  \tau=\lambda^{-1/\gamma},
  \qquad
  K=\min\{T,\lfloor\tau\rfloor\}.
\]
Then 
\[
  A_\lambda(T)
  =\lambda^{-1}\sum_{t=1}^{K}t^{-\gamma}+(T-K).
  \label{eq:pht-powerlaw-exact-sum}
\]
For $0<\gamma<1$, when \(T\le \tau\), by 
\[ \int_1^{K+1}x^{-\gamma}\,dx
\le
\sum_{t=1}^K t^{-\gamma}
\le
1+\int_1^K x^{-\gamma}\,dx \, , \]
we get that 
\[A_\lambda(T) = \Theta_\gamma(\lambda^{-1}T^{1-\gamma})\, , \]
where \(\Theta_\gamma(\cdot )\) subsumes constants that only depend on \(\gamma\). When \(T\ge \tau \), we have 
\[ A_\lambda(T) = \lambda^{-1} \sum_{t\le \tau} t^{-\gamma} + (T - \tau) +O(1) = \Theta_\gamma(T) \, .\]
If $\lambda\ge1$, then $p_t=1$ for every
$t\ge1$ and $A_\lambda(T)=T$ exactly.

To determine the case boundaries, first consider when \(T \le \tau \) the solution to $A_\lambda(T)=\Theta(N)$ is
\[
  T_0=(\lambda N)^{1/(1-\gamma)}.
  \label{eq:pht-powerlaw-Tzero}
\]
Hence \(T_0=\sqrt N\) whenever \(
  \lambda=N^{-(1+\gamma)/2}\). 
On the other hand when \( \lambda=N^{-\gamma} \), \(T_0=\tau=N\).  Combining these facts proves all three lower bounds in
Theorem~\ref{thm:powerlaw}. 

\paragraph{Matching block-Grover upper bound.}
We use a standard Grover search subroutine  \(\grover\) with the following promise~\cite{G96,BBHT98}. There is an absolute constant \(c\) such that, for every Boolean function on \(m\) inputs with at least one marked item, \(\grover\) uses at most \( c\sqrt m \) oracle queries which outputs a marked item with probability at least \(1/4\), without knowing the number of marked items. 

Let \(S\) denote the target complexity in the respective lower-bound cases. 
\[
S=S(N,\lambda,\gamma):=
\begin{cases}
\sqrt N,
& \lambda\le N^{-(1+\gamma)/2},\\[1mm]
(\lambda N)^{1/(1-\gamma)},
& N^{-(1+\gamma)/2}\le \lambda\le N^{-\gamma},\\[1mm]
N,
& \lambda\ge N^{-\gamma}.
\end{cases}
\]
We show that $\OR_N$ can be computed with $O_{\gamma}(S)$ queries.

Fix an absolute constant $\kappa>0$. Let
\(H :=\kappa S \) denote a total query budget and denote \(\bar p:=p_{H}
=\min\{1,\lambda Q_*^\gamma\}\). Choose \(m=\Theta\!\left(\min\{N,\bar p^{-2}\}\right)\) and some univeral constant \(R>0\).

\begin{algorithm}[ht]
\caption{Search under aging noise}
\label{alg:powerlaw-search}
\begin{algorithmic}[1]
\State Partition $[N]$ into blocks $B_1,\ldots,B_K$ of size at most $m$
\For{$j=1,\ldots,K$}
    \For{$r=1,\ldots,R$}
        \State $i\gets\Call{\grover}{B_j}$
        \Comment{at most $c\sqrt m$ queries}
        \State Query $x_i$ and \Return $1$ if $x_i=1$
    \EndFor
\EndFor
\State \Return $0$
\end{algorithmic}
\end{algorithm}

Since $K=O(N/m)$ and $R=O(1)$, the total query count is
\[
O\!\left(\frac{N}{m}\sqrt m\right)
=
O\!\left(\frac{N}{\sqrt m}\right)
=
O\!\left(\max\{\sqrt N,N\bar p\}\right).
\]
We claim that
\[
N\bar p=O_\gamma(S).
\]

In the regime of $S=\sqrt N$ and
$\lambda\le N^{-(1+\gamma)/2}$, we have
\[
N\bar p
\le
N\lambda(\kappa S)^\gamma
=
O_\gamma\!\left(
N\,N^{-(1+\gamma)/2}N^{\gamma/2}
\right)
=
O_\gamma(\sqrt N)
=
O_\gamma(S).
\]

In the intermediate regime,
\[
S=(\lambda N)^{1/(1-\gamma)},
\qquad\text{hence}\qquad
\lambda N=S^{1-\gamma}.
\]
Therefore
\[
N\bar p
\le
N\lambda(\kappa S)^\gamma
=
O_\gamma(N\lambda S^\gamma)
=
O_\gamma(S).
\]

In the classical regime $S=N$, trivially
\[
N\bar p\le N=S.
\]
Thus the total number of queries is $O_\gamma(S)$.

It remains to check that a marked block is found with constant
probability.  One invocation of $\Call{GroverUnknown}{}$ uses at most
$c\sqrt m$ query slots, all before query $Q_*$, and therefore each has
dephasing probability at most $\bar p$.  Hence its probability of
experiencing no dephasing is at least
\[
(1-\bar p)^{c\sqrt m}.
\]
By the choice
\[
m=\Theta\!\left(\min\{N,\bar p^{-2}\}\right),
\]
we have
\[
\bar p\sqrt m=O(1),
\]
so this probability is bounded below by a positive constant.
Conditioned on this clean event, the Grover subroutine succeeds with
constant probability whenever the block is nonempty.  A constant number
of repetitions therefore gives bounded error, while the verification
query rules out false positives.

Combining with the lower bound,
\[
Q_\varepsilon^{(p_t)}(\OR_N)
=
\Theta_{\varepsilon,\gamma}(S) \, .
\]

\paragraph{Example: an intermediate \(N^{2/3}\) regime.}
Take $\gamma=1/2$ and $\lambda=N^{-2/3}$.  Then
\[
  T_0=(\lambda N)^2=N^{2/3},
  \qquad
  \tau=\lambda^{-2}=N^{4/3} \, .
\]
At
$H=\Theta(N^{2/3})$,
\[
  p_H=\Theta(N^{-1/3}),
  \qquad
  m=\Theta(N^{2/3}).
\]
There are $\Theta(N^{1/3})$ blocks, each searched in
$\Theta(N^{1/3})$ queries, and therefore
\[
  Q_\eps^{(p_t)}(\OR_N)=\Theta_\eps(N^{2/3}).
\]

\subsection{Proving the computational-opportunity theorem}
\label{sec:proof-cot}


Thanks to the IID analysis, we can use the same purified weighted-star setup, hidden index record, and \(A/Q/C\) decomposition. Adopting the same potential \(\Psi_t\), we can show the one-step growth and terminal condition in Lemmas~\ref{lem:ext-dephase-potential} and~\ref{lem:ext-terminal} essentially unchanged except for the time-dependent dephasing probability.

For~\cref{lem:ext-dephase-potential}, conditioned on the accessible history \(\cF_{t-1}\), the \(t\)-th query is dephased with conditional probability \(p_t\) rather than a fixed probability \(p\). Since the local \(A/Q/C\) decomposition and the active magnitude \(a_t\) are defined after conditioning on this history, averaging the coherent and dephased branches gives exactly the same one-query bound with \(p\) replaced by \(p_t\).

The terminal condition~\cref{lem:ext-terminal} is unchanged, because it depends only on the final record-dependent decomposition and the algorithm's worst-case correctness, not on how the dephasing probabilities vary over time. IID proof hence carries over verbatim.

\begin{lemma}[Conditional one-query growth]
\label{lem:ext-dephase-potential}
Conditioned on the accessible history immediately before query~$t$,
\[
  \E[\Psi_t-\Psi_{t-1}\mid\cF_{t-1}]
  \le
  \frac{12}{W}
  +
  \frac{12a_t}{\sqrt W}
  -2p_ta_t^2.
\]
\end{lemma}

\begin{lemma}[Record-dependent terminal progress]
\label{lem:ext-terminal}
If the algorithm has worst-case error at most $\eps<1/2$, then
\[
  \E\Psi_T\ge\gamma_\eps,
  \qquad
  \gamma_\eps
  :=\frac{\bigl(1-2\sqrt{\eps(1-\eps)}\bigr)^2}{36}>0.
\]
\end{lemma}

We can then put them together to prove~\cref{thm:opportunity}. 
\begin{proof}
Assume $W>0$ and fix an arbitrary
finite-cost predictable opportunity-allocation process $\theta$.  Let
\[
  S=\sum_{t=1}^T\theta_t,
  \qquad
  J=\sum_{t=1}^T\frac{(1-\theta_t)^2}{p_t}.
\]
Although expectations involving $\Psi_t$ also average over the hard input
prior and the algorithm, the variables $S$ and $J$ depend only on the schedule.
Consequently their expectations are exactly their expectations under $\mu$.
Split the linear term in Lemma~\ref{lem:local-potential} as
\[
  12a_t=12\theta_ta_t+12(1-\theta_t)a_t.
\]
For $p_t>0$, we get the inequality below which saturates at \( a_t = \frac{3(1-\theta_t)^2}{p_t\sqrt W}\,,\)  
\[
  \frac{12(1-\theta_t)a_t}{\sqrt W}-2p_ta_t^2
  \le
  \frac{18}{W}\frac{(1-\theta_t)^2}{p_t}.
  \label{eq:pht-square-completion}
\]
At $p_t=0$, finite cost forces $\theta_t=1$, and both exposed terms vanish.
Using $a_t\le1$, summing the local inequalities, and applying the tower
property yields
\[
  \E\Psi_T
  \le
  \frac{12T}{W}
  +\frac{12}{\sqrt W}\E S
  +\frac{18}{W}\E J.
  \label{eq:pht-summed-drift}
\]

Since $p_t\le1$ and $0\le\theta_t\le1$, pathwise
\[
  S^2+J
  \ge
  \sum_{t=1}^T
    \bigl(\theta_t^2+(1-\theta_t)^2\bigr)
  \ge\frac T2.
\]
Therefore
\[
  T\le2r_\theta(\mu),
  \qquad
  \E J\le r_\theta(\mu),
  \qquad
  \E S\le\sqrt{\E S^2}\le\sqrt{r_\theta(\mu)}.
\]
Substitution in \eqref{eq:pht-summed-drift} gives
\[
  \E\Psi_T
  \le
  12\sqrt{\frac{r_\theta(\mu)}{W}}
  +42\frac{r_\theta(\mu)}{W}.
  \label{eq:pht-resource-upper}
\]

Let $r=r_\theta(\mu)/W$.  Lemma~\ref{lem:terminal} implies
\[
  \gamma_\eps\le12\sqrt r+42r,
\]
and hence
\[
  r\ge
  c_\eps
  :=\left(
      \frac{\sqrt{144+168\gamma_\eps}-12}{84}
    \right)^2>0.
\]
Thus $r_\theta(\mu)\ge c_\eps W$ for every predictable $\theta$.
Taking the infimum only now proves
$\cR(\mu)\ge c_\eps W$.  Here \(W\) is the \(\fbs\) centered at $0^n$. As we argued in~\cref{sec:iid}, for an
arbitrary center $x^\star$, apply it to the translated function
\[
  g(y):=f(x^\star\oplus y).
\]
Then
\[
  \fbs(g,0^n)=\fbs(f,x^\star),
\]
and querying $g$ is equivalent to querying $f$ after the fixed permutation
$y\mapsto x^\star\oplus y$ of the query-index basis.  Query-index dephasing is
invariant under such a permutation, so the same schedule law and resource
apply.  Maximizing over $x^\star$ and over the fractional packing therefore
gives $\fbs(f)$. 
\end{proof}

\ifsubmission\else
\subsection*{Acknowledgments.}
The authors were partially supported by the National Science Foundation under Award No. 2348489 
\fi

\section*{AI Disclosure}
In addition to checking typos, grammars, and editorial suggestions and corrections by LLMs, there are two substantively interactions.

\begin{itemize}
\item The authors asked Claude Fable 5 to investigate whether the \(\fbs(f,0^n)\) lower bound for negligent queries could be extended to general \(\fbs(f)\).  The model suggested a partial-function counterexample.  The authors independently verified the construction and subsequently generalized it to the lifting result in~\cref{thm:neg-partial-lifting}.

\item The authors proposed the extended stochastic dephasing model and a conjectured generalization of the hybrid lower bound to GPT-5.6 Sol (High mode).  The model identified a counterexample to the conjectured formulation and suggested the opportunity-resource formulation used in this paper.  The authors independently formulated and proved the Computational-Opportunity Theorem, and developed its applications, including the aging-dephasing example.

\end{itemize}
The authors verified the correctness and originality of all AI-assisted content and take full responsibility for the results in this paper.

\newpage
\bibliographystyle{alpha}
\bibliography{qfbs}
\end{document}